\documentclass[english]{article}
\usepackage{amsmath}
\usepackage{amssymb}

\usepackage{cite}

\usepackage{subfig}

\usepackage{color}
\usepackage{tikz}
\usetikzlibrary{positioning}
\usetikzlibrary{shapes}
\usetikzlibrary{patterns}
\usetikzlibrary{arrows}
\usetikzlibrary{calc}
\tikzset{every picture/.style={>=stealth'}}

\usepackage{graphicx}
\usepackage{pgfplots}
\pgfplotsset{compat=1.16}

\makeatletter
\newenvironment{customlegend}[1][]{%
    \begingroup
    \pgfplots@init@cleared@structures
    \pgfplotsset{#1}%
}{%
    \pgfplots@createlegend
    \endgroup
}%

\def\addlegendimage{\pgfplots@addlegendimage}
\makeatother

\colorlet{1level}{cyan}
\colorlet{2level}{orange}

\tikzstyle{dots_0nodes}=[only marks, mark=o, mark size=2.5pt,
                         mark options={color=black, opacity=0.5}]
\tikzstyle{dots_1level}=[only marks, mark=*, mark size=2.5pt,
                         mark options={color=1level, mark color=white, fill=1level, opacity=0.5}]
\tikzstyle{dots_2level}=[only marks, mark=*, mark size=2.5pt,
                         mark options={color=2level, fill=2level, opacity=0.5}]

\tikzstyle{plot_adiar}=[color=black, mark=o, mark size=1pt, line width=0.7pt]
\tikzstyle{plot_buddy}=[color=red, mark=triangle, mark size=1pt, line width=0.7pt]
\tikzstyle{plot_cudd}=[color=blue, mark=diamond, mark size=1pt, line width=0.7pt]
\tikzstyle{plot_sylvan}=[color=purple, mark=square, mark size=1pt, line width=0.7pt]

\tikzstyle{plot_adiar_v120}=[color=black, densely dashed, mark=o, mark size=1pt, line width=0.7pt]

\tikzstyle{plot_0nodes}=[color=black, dashed, line width=0.7pt]
\tikzstyle{plot_1level}=[color=1level, line width=0.7pt]
\tikzstyle{plot_2level}=[color=2level, line width=0.7pt]

\usepackage{hyperref,doi}
\usepackage{multirow}
\hypersetup{
    colorlinks,
    linkcolor={red!50!black},
    citecolor={green!40!black},
    urlcolor={blue!40!black}
}

\newcommand{\B}[0]{\ensuremath{\mathbb{B}}}
\newcommand{\N}[0]{\ensuremath{\mathbb{N}}}
\newcommand{\Lcal}[0]{\ensuremath{\mathcal{L}}}
\newcommand{\LcalBDD}[0]{\ensuremath{\mathcal{L}_{\mathit{OBDD}}}}

\newcommand{\labelof}[0]{\ensuremath{\mathit{label}}}

\newcommand{\inset}[0]{\ensuremath{\mathit{in}}}
\newcommand{\outset}[0]{\ensuremath{\mathit{out}}}

\newcommand{\abs}[1]{\left\lvert #1 \right\rvert}

\newcommand{\Adiar}[0]{Adiar}
\newcommand{\TPIE}[0]{TPIE}

\newcommand{\arc}[3]{\ensuremath{#1 \xrightarrow{_{#2}} #3}}

\newcommand{\weightedPPR}[0]{\ensuremath{\mathit{PPR}}}

\def\orcidID#1{${}^{\smash{\href{https://orcid.org/#1}{\protect\raisebox{-1.25pt}{\protect\includegraphics{figure/orcid_color-eps-converted-to}}}}}$}

\usepackage{amsthm}

\newtheorem{theorem}{Theorem}
\newtheorem{definition}[theorem]{Definition}
\newtheorem{lemma}[theorem]{Lemma}
\newtheorem{proposition}[theorem]{Proposition}
\newtheorem{corollary}[theorem]{Corollary}
\newtheorem*{observation*}{Observation}

\newcommand*{\mailto}[1]{\href{mailto:#1}{\nolinkurl{#1}}}

\def\arxiv{1}

\begin{document}

\title{Predicting Memory Demands of BDD Operations\\using Maximum Graph Cuts
  \\ {\large (\emph{Extended Version})}}

\author{Steffan Christ Sølvsten
  \\ \\ {\normalsize Aarhus University, Denmark}
  \and Jaco van de Pol
  \\ \\ {\normalsize \{%
  \href{mailto:soelvsten@cs.au.dk}{\color{black} soelvsten},%
  \href{mailto:jaco@cs.au.dk}{\color{black} jaco}%
  \}@cs.au.dk}
  }

% typeset the header of the contribution
\maketitle

% ---------------------------------------------------------------------------- %

\begin{abstract}\noindent
  The BDD package Adiar manipulates Binary Decision Diagrams (BDDs) in external
  memory. This enables handling big BDDs, but the performance suffers when
  dealing with moderate-sized BDDs. This is mostly due to initializing expensive
  external memory data structures, even if their contents can fit entirely
  inside internal memory.

  The contents of these auxiliary data structures always correspond to a graph
  cut in an input or output BDD. Specifically, these cuts respect the levels of
  the BDD. We formalise the shape of these cuts and prove sound upper bounds on
  their maximum size for each BDD operation.

  We have implemented these upper bounds within Adiar. With these bounds, it can
  predict whether a faster internal memory variant of the auxiliary data
  structures can be used. In practice, this improves Adiar's running time across
  the board. Specifically for the moderate-sized BDDs, this results in an
  average reduction of the computation time by $86.1\%$ (median of $89.7\%$). In
  some cases, the difference is even $99.9\%$. When checking equivalence of
  hardware circuits from the EPFL Benchmark Suite, for one of the instances the
  time was decreased by $52$ hours.
\end{abstract}

% ---------------------------------------------------------------------------- %

\section{Introduction} \label{sec:introduction}

A Binary Decision Diagrams (BDD) \cite{Bryant1986} is a data structure that has
found great use within the field of combinatorial logic and verification. Its
ability to concisely represent and manipulate Boolean formulae is the key to
many symbolic model checkers, e.g.\
\cite{Cimatti2000,Gammie2004,Ciardo2009,Kant2015,Lomuscio2017,He2020,Amparore2022}
and recent symbolic synthesis algorithms \cite{Yi2022}. Bryant and Heule
recently found a use for BDDs to create SAT and QBF solvers with certification
capabilities \cite{Bryant2021:SAT,Bryant2021:QBF,Bryant2022} that are better at
proof generation than conventional SAT solvers.

Adiar~\cite{Soelvsten2022:TACAS} is a redesign of the classical BDD algorithms
such that they are optimal in the I/O model of Aggarwal and
Vitter~\cite{Aggarwal1987}, based on ideas from Lars Arge~\cite{Arge1995}. As
shown in Fig.~\ref{fig:motivation}, this enables Adiar to handle BDDs beyond the
limits of main memory with only a minor slowdown in performance, unlike
conventional BDD implementations. Adiar is implemented on top of the \TPIE\
library~\cite{Vengroff1994,Molhave2012}, which provides external memory sorting
algorithms, file access, and priority queues, while making management of I/O
transparent to the programmer. These external memory data structures work by
loading one or more blocks from files on disk into internal memory and
manipulating the elements within these blocks before storing them again on the
disk. Their I/O-efficiency stems from a carefully designed order in which these
blocks are retrieved, manipulated, and stored. Yet, initializing the internal
memory in preparation to do so is itself costly -- especially if purely using
internal memory would have sufficed. This is evident in
Fig.~\ref{fig:motivation} (cf.~Section~\ref{sec:experiments:impact} for more
details) where Adiar's performance is several orders of magnitude worse than
conventional BDD packages for smaller instance sizes. In fact, Adiar's
performance decreases when the amount of internal memory increases.

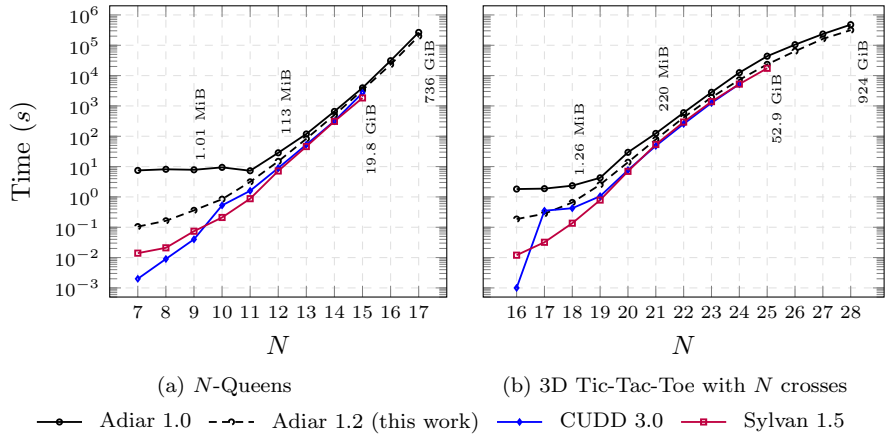
\begin{figure}[t]
  \centering

  \subfloat[$N$-Queens]{
    \begin{tikzpicture}
      \begin{axis}[%
        width=0.5\linewidth, height=0.45\linewidth,
        every tick label/.append style={font=\scriptsize},
        % x-axis
        xlabel={$N$},
        xmajorgrids=true,
        xtick={7,8,9,10,11,12,13,14,15,16,17},
        % y-axis
        ylabel={Time ($s$)},
        ytick distance={10},
        ymode=log,
        yminorgrids=false,
        ymajorgrids=true,
        ymin=0.0005,
        ymax=2000000,
        grid style={dashed,black!12},
        ]
        \addplot+ [style=plot_adiar]
        table {./data/motivation_queens_adiar_v100.tex};

        \addplot+ [style=plot_adiar_v120]
        table {./data/motivation_queens_adiar_v120.tex};

        % \addplot+ [style=plot_buddy]
        % table {./data/motivation_queens_buddy.tex};

        \addplot+ [style=plot_cudd]
        table {./data/motivation_queens_cudd.tex};

        \addplot+ [style=plot_sylvan]
        table {./data/motivation_queens_sylvan.tex};
      \end{axis}

      \node[color=black, rotate=90] at (1.2, 2.3) {\tiny $1.01$~MiB};
      \node[color=black, rotate=90] at (2.34, 2.6) {\tiny $113$~MiB};
      \node[color=black, rotate=90] at (3.47, 2.1) {\tiny $19.8$~GiB};
      \node[color=black, rotate=90] at (4.25, 3.0) {\tiny $736$~GiB};
    \end{tikzpicture}
  }
  \subfloat[3D Tic-Tac-Toe with $N$ crosses]{
    \begin{tikzpicture}
      \begin{axis}[%
        width=0.57\linewidth, height=0.45\linewidth,
        every tick label/.append style={font=\scriptsize},
        % x-axis
        xlabel={$N$},
        xmajorgrids=true,
        xtick={16,17,18,19,20,21,22,23,24,25,26,27,28},
        % y-axis
        ytick distance={10},
        ymode=log,
        yminorgrids=false,
        ymajorgrids=true,
        ymin=0.0005,
        ymax=2000000,
        grid style={dashed,black!12},
        yticklabels={,,}
        ]
        \addplot+ [style=plot_adiar]
        table {./data/motivation_tic_tac_toe_adiar_v100.tex};

        \addplot+ [style=plot_adiar_v120]
        table {./data/motivation_tic_tac_toe_adiar_v120.tex};

        % \addplot+ [style=plot_buddy]
        % table {./data/motivation_tic_tac_toe_buddy.tex};

        \addplot+ [style=plot_cudd]
        table {./data/motivation_tic_tac_toe_cudd.tex};

        \addplot+ [style=plot_sylvan]
        table {./data/motivation_tic_tac_toe_sylvan.tex};
      \end{axis}

      \node[color=black, rotate=90] at (1.28, 2.1) {\tiny $1.26$~MiB};
      \node[color=black, rotate=90] at (2.4, 2.9) {\tiny $220$~MiB};
      \node[color=black, rotate=90] at (3.9, 2.5) {\tiny $52.9$~GiB};
      \node[color=black, rotate=90] at (5.05, 3.0) {\tiny $924$~GiB};
    \end{tikzpicture}
  }

  \begin{tikzpicture}
    \begin{customlegend}[
        legend columns=-1,
        legend style={draw=none,column sep=1ex},
        legend entries={
          \footnotesize \Adiar~1.0,
          \footnotesize \Adiar~1.2 (this work),
          % \footnotesize BuDDy~2.4,
          \footnotesize CUDD~3.0,
          \footnotesize Sylvan~1.5
        }
      ]
      \addlegendimage{style=plot_adiar}
      \addlegendimage{style=plot_adiar_v120}
      % \addlegendimage{style=plot_buddy}
      \addlegendimage{style=plot_cudd}
      \addlegendimage{style=plot_sylvan}
    \end{customlegend}
  \end{tikzpicture}

  \caption{Running time solving combinatorial BDD benchmarks. Some instances are
    labelled with the size of the largest BDD constructed to solve them.}
  \label{fig:motivation}
\end{figure}

This shortcoming is not desirable for a BDD package: while our research focuses
on enabling large-scale BDD manipulation, end users should not have to consider
whether their BDDs will be large enough to benefit from Adiar. Solving this also
paves the way for Adiar to include complex BDD operations where conventional
implementations recurse on intermediate results, e.g.\ \emph{Multi-variable
  Quantification}, \emph{Relational Product}, and \emph{Variable Reordering}. To
implement the same, Adiar has to run multiple sweeps. Yet, each of these sweeps
suffer when they unecessarily use external memory data structures. Hence, it is
vital to overcome this shortcomming, to ensure that an I/O-efficient
implementations of these complex BDD operations will also be usable in practice.

The linearithmic I/O- and time-complexity of Adiar's algorithms also applies to
the lower levels of the memory hierarchy, i.e.\ between the cache and RAM.
Hence, there is no reason to believe that the bad performance for smaller
instances is inherently due to the algorithms themselves; if they used an
internal memory variant of all auxiliary data structures, then Adiar ought to
perform well for much smaller instances. In fact, this begs the question: while
we have investigated the applicability of these algorithms at a large-scale in
\cite{Soelvsten2022:TACAS}, how can they seamlessly handle both small and large
BDDs efficiently?

We argue that simple solutions are unsatisfactory: A first idea would be to
start running classical, depth-first BDD algorithms until main memory is
exhausted. In that case, the computation is aborted and restarted with external
memory algorithms. But, this strategy doubles the running time. While it would
work well for small instances, the slowdown for large instances would be
unacceptable. Alternatively, both variants could be run in parallel. But, this
would halve the amount of available memory and again slow down large instances.

A second idea would be to start running Adiar's I/O-efficient algorithms with an
implementation of all auxiliary data structures in internal memory. In this
case, if memory is exhausted, the data coudl be copied to disk, and the
computation could be resumed with external memory. This could be implemented
neatly with the \emph{state pattern}: a wrapper switches transparently to the
external memory variant when needed. Yet, moving elements from one sorted data
structure to another requires at least linear time. Even worse, such a wrapper
adds an expensive level of indirection and hinders the compiler in inlining and
optimising, since the actual data structure is unknown at compile-time.

Instead, we propose to use the faster, internal-memory version of Adiar's
algorithms only when it is guaranteed to succeed. This avoids re-computations,
duplicate storage, as well as the costs of indirection. The main research
question is how to predict a sound upper bound on the memory required for a BDD
operation, and what information to store to compute these bounds efficiently.

\subsection{Contributions}

In Section~\ref{sec:cut}, we introduce the notion of an $i$-level cut for
Directed Acyclic Graphs (DAGs). Essentially, the shape of these cuts is
constricted to span at most $i$ levels of the given DAG. Previous results in
\cite{Lampis2011} show that for $i \geq 4$ the problem of computing the maximum
$i$-level cut is NP-complete. We show that for $i \in \{ 1,2 \}$ this problem is
still computable in polynomial time. These polynomial-time algorithms can be
implemented using a linearithmic amount of time and I/Os. But instead, we use
over-approximations of these cuts. As described in Section~\ref{sec:cut:adiar},
their computation can be piggybacked on existing BDD algorithms, which is
considerably cheaper: for $1$-level cuts, this only adds a $1\%$ linear-time
overhead and does not increase the number of I/O operations.

Investigating the structure of BDDs from the perspective of $i$-level cuts for
$i \in \{ 1,2 \}$ in Section~\ref{sec:cut:bdd} and \ref{sec:cut:terminal}, we
obtain sound upper bounds on the maximum $i$-level cuts of a BDD operation's
output, purely based on the maximum $i$-level cut of its inputs. Using these
upper bounds, Adiar can decide in constant time whether to run the next
algorithm with internal or external memory data structures. Here, only one
variant is run, all memory is dedicated to it, and the exact type of the
auxiliary data structures are available to the compiler.

Our experiments in Section~\ref{sec:experiments} show that it is a good strategy
to compute the 1-level cuts, and to use them to infer an upper bound on the
2-level cuts. This strategy is sufficient to address Adiar's performance issues
for the moderate-sized instances while also requiring the least computational
overhead. As Fig.~\ref{fig:motivation} shows, adding these cuts to Adiar with
version 1.2 removes the overhead introduced by initializing \TPIE's external
memory data structures and so greatly improves Adiar's performance. For example,
to verify the correctness of the small and moderate instances of the EPFL
combinational benchmark circuits~\cite{Amaru2015}, the use of $i$-level cuts
decreases the running time from $56.5$ hours down to $4.0$ hours.

\section{Preliminaries} \label{sec:preliminaries}

\subsection{Graph and Cuts}

A directed graph is a tuple $(V, A)$ where $V$ is a finite set of vertices and
$A \subseteq V \times V$ a set of arcs between vertices. The set of incoming
arcs to a vertex $v \in V$ is $\inset(v) = A \cap (V \times \{v\})$, its
outgoing arcs are $\outset(v) = A \cap (\{v\} \times V)$, and $v$ is a
\emph{source} if its indegree $\abs{\inset(v)} = 0$ and a \emph{sink} if its
outdegree $\abs{\outset(v)} = 0$.

A cut of a directed graph $(V,A)$ is a partitioning $(S,T)$ of $V$ such that $S
\cup T = V$ and $S \cap T = \emptyset$. Given a weight function $w : A
\rightarrow \mathbb{R}$ the weighted maximum cut problem is to find a cut
$(S,T)$ such that $\sum_{a \in S \times T \cap A} w(a)$ is maximal, i.e.\ where
the total weight of arcs crossing from some vertex in $S$ to one in $T$ is
maximised. Without decreasing the weight of a cut, one may assume that all
sources in $V$ are part of the partition $S$ and all sinks are part of $T$. The
maximum cut problem is NP-complete for directed graphs \cite{Papadimitriou1991}
and restricting the problem to directed acyclic graphs (DAGs) does not decrease
the problem's complexity \cite{Lampis2011}.

If the weight function $w$ merely counts the number of arcs that cross a cut,
i.e.\ $\forall a \in A : w(a) = 1$, the problem above reduces to the
\emph{unweighted} maximum cut problem where a cut's weight and size are
interchangeable.

\subsection{Binary Decision Diagrams} \label{sec:preliminaries:bdd}

A Binary Decision Diagram (BDD)~\cite{Bryant1986}, as depicted in
Fig.~\ref{fig:bdd_example}, is a DAG $(V,A)$ that represents an $n$-ary Boolean
function. It has a single source vertex $r \in V$, usually referred to as the
\emph{root}, and up to two sinks for the Boolean values $\B = \{ \bot, \top \}$,
usually referred to as \emph{terminals} or \emph{leaves}. Each non-sink vertex
$v \in V \setminus \B$ is referred to as a BDD \emph{node} and is associated
with an input variable $x_i \in \{ x_0, x_1, \dots, x_{n-1} \}$ where
$\labelof(v) = i$. Each arc is associated with a Boolean value, i.e.\ $A
\subseteq V \times \B \times V$ (written as $\arc{v}{b}{v'}$ for a $(v,b,v') \in
A$), such that each BDD node $v$ represents a binary choice on its input
variable. That is, $\outset(v) = \{ \arc{v}{\bot}{v'}, \arc{v}{\top}{v''} \}$,
reflecting $x_i$ being assigned the value $\bot$, resp.\ $\top$. Here, $v'$ is
said to be $v$'s \emph{low} child while $v''$ is its \emph{high} child.

\if\arxiv1
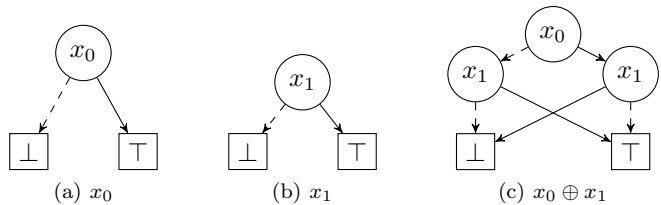
\begin{figure}[t]
\else
\begin{figure}[t]
\fi
  \centering

  \subfloat[$x_0$]{
    \label{fig:bdd_example:x0}

    \begin{tikzpicture}
      % nodes
      \node[shape = circle, draw = black]                                    (0) {$x_0$};

      % leafs
      \node[shape = rectangle, draw = black, below left=.8cm and .2cm of 0]   (sink_F) {$\bot$};
      \node[shape = rectangle, draw = black, below right=.8cm and .2cm of 0]  (sink_T) {$\top$};

      % arcs
      \draw[->]
      (0) edge (sink_T)
      ;

      \draw[->, dashed]
      (0) edge (sink_F)
      ;
    \end{tikzpicture}
  }
  \qquad
  \subfloat[$x_1$]{
    \label{fig:bdd_example:x1}

    \begin{tikzpicture}
      % nodes
      \node[shape = circle, draw = black]                                    (0) {$x_1$};

      % leafs
      \node[shape = rectangle, draw = black, below left=.41cm and .2cm of 0]   (sink_F) {$\bot$};
      \node[shape = rectangle, draw = black, below right=.41cm and .2cm of 0]  (sink_T) {$\top$};

      % arcs
      \draw[->]
      (0) edge (sink_T)
      ;

      \draw[->, dashed]
      (0) edge (sink_F)
      ;
    \end{tikzpicture}
  }
  \qquad
  \subfloat[$x_0 \oplus x_1$]{
    \label{fig:bdd_example:xor}

    \begin{tikzpicture}
      % nodes
      \node[shape = circle, draw = black]                                (0) {$x_0$};
      \node[shape = circle, draw = black, below left= 0cm and .5cm of 0] (1) {$x_1$};
      \node[shape = circle, draw = black, below right=0cm and .5cm of 0] (2) {$x_1$};

      % leafs
      \node[shape = rectangle, draw = black, below=.41cm of 1]  (sink_F) {$\bot$};
      \node[shape = rectangle, draw = black, below=.41cm of 2]  (sink_T) {$\top$};

      % arcs
      \draw[->]
      (0) edge (2)
      (1) edge (sink_T)
      (2) edge (sink_F)
      ;

      \draw[->, dashed]
      (0) edge (1)
      (1) edge (sink_F)
      (2) edge (sink_T)
      ;
    \end{tikzpicture}
  }
  \caption{Examples of Reduced Ordered Binary Decision Diagrams. Terminals are
    drawn as boxes with the Boolean value and BDD nodes as circles with the
    decision variable. \emph{Low} edges are drawn dashed while \emph{high} edges
    are solid.}
  \label{fig:bdd_example}
\end{figure}

An Ordered Binary Decision Diagram (OBDD) restricts the DAG such that all paths
follow some total variable ordering $\pi$: for every arc $\arc{v_1}{}{v_2}$
between two distinct nodes $v_1$ and $v_2$, $\labelof(v_1)$ must precede
$\labelof(v_2)$ according to the order $\pi$. A Reduced Ordered Binary Decision
Diagram (ROBDD) further adds the restriction that for each node $v$ where
$\outset(v) = \{ \arc{v}{\bot}{v'}, \arc{v}{\top}{v''} \}$, (1) $v' \neq v''$
and (2) there exists no other node $u \in V$ such that $\labelof(v) =
\labelof(u)$ and $\outset(u) = \{ \arc{u}{\bot}{v'}, \arc{u}{\top}{v''} \}$. The
first requirement removes \emph{don't care} nodes while the second removes
\emph{duplicates}. Assuming a fixed variable ordering $\pi$, an ROBDD is a
canonical representation of the Boolean function it represents
\cite{Bryant1986}. Without loss of generality, we will assume $\pi$ is the
identity.

This graph-based representation allows one to indirectly manipulate Boolean
formulae by instead manipulating the corresponding DAGs. For simplicity, we will
focus on the Apply operation in this paper, but our results can be generalised
to other operations. Apply computes the ROBDD for $f \odot g$ given ROBDDs for
$f$ and $g$ and a binary operator $\odot~:~\B \times \B \rightarrow \B$. This is
done with a product construction of the two DAGs, starting from the pair $(r_f,
r_g)$ of the roots of $f$ and $g$. If terminals $b_f$ from $f$ and $b_g$ from
$g$ are paired then the resulting terminal is $b_f \odot b_g$. Otherwise, when
nodes $v_f$ from $f$ and $v_g$ from $g$ are paired, a new BDD node is created
with label $\ell = \min(\labelof(v_f), \labelof(v_g))$, and its low and high
child are computed recursively from pairs $(v_f', v_g')$. For the low child,
$v_f'$ is $v_f.\mathit{low}$ if $\labelof(v_f) = \ell$ and $v_f$ otherwise;
$v_g'$ is defined symmetrically. The recursive tuple for the high child is
defined similarly.

\subsubsection{Zero-suppressed Decision Diagrams} \label{sec:preliminaries:zdd}

A Zero-suppressed Decision Diagram (ZDD)~\cite{Minato1993} is a variation of
BDDs where the first reduction rule is changed: a node $v$ for the variable
$\labelof(v)$ with $\outset(v) = \{ \arc{v}{\bot}{v'}, \arc{v}{\top}{v''} \}$ is
not suppressed if $v$ is a \emph{don't care} node, i.e.\ if $v' = v''$, but
rather if it assigns the variable $\labelof(v)$ to $\bot$, i.e.\ if $v'' =
\bot$. This makes ZDDs a better choice in practice than BDDs to represent
functions $f$ where its on-set, $\{ \vec{x} \mid f(\vec{x}) = \top \}$, is
sparse.

The basic notions behind the BDD algorithms persist when translated to ZDDs, but
it is important for correctness that the ZDD operations account for the shape of
the suppressed nodes. For example, the \emph{union} operation needs to replace
recursion requests for $(v_f,v_g)$ with $(v_f,\bot)$ if $\labelof(v_f) <
\labelof(v_g)$ and with $(\bot,v_g)$ if $\labelof(v_f) > \labelof(v_g)$.

\subsubsection{Levelised Algorithms in Adiar} \label{sec:preliminaries:adiar}

BDDs and ZDDs are usually manipulated with recursive algorithms that use two
hash tables: one for memoisation and another to enforce the two reduction rules
\cite{Brace1990,Minato1990}. Lars Arge noted in \cite{Arge1995,Arge1996} that
this approach is not efficient in the I/O-model of Aggarwal and
Vitter~\cite{Aggarwal1987}. He proposed to address this issue by processing all
BDDs iteratively level by level with the time-forward processing technique
\cite{Chiang1995,Meyer2003}: recursive calls are not executed at the time of
issuing the request but are instead deferred with one or more priority queues
until the necessary elements are encountered in the inputs. In
\cite{Soelvsten2022:TACAS}, we implemented this approach in the BDD package
Adiar. Furthermore, with version 1.1 we have extended this approach to ZDDs
\cite{Soelvsten2023:NFM}.

In Adiar, each decision diagram is represented as a sequence of its BDD nodes.
Each BDD node is uniquely identifiable by the pair $(\ell, i)$ of its level
$\ell$, i.e.\ its variable label, and its level-index $i$. And so, each BDD node
can be represented as a triple of its own and its two children's unique
identifiers (uids). The entire sequence of BDD nodes follows a level by level
ordering of nodes which is equivalent to a lexicographical sorting on their uid.
For example, the three BDDs in Fig.~\ref{fig:bdd_example} are stored on disk as
the lists in Fig.~\ref{fig:bdd_topological}.

\if\arxiv1
\begin{figure}[ht!]
\else
\begin{figure}[t]
\fi
  \centering

  \begin{tabular}{r c l c l c l c}
    \ref{fig:bdd_example:x0}:
    & [ & $((0, 0), \bot, \top)$ & ]
    \\
    \ref{fig:bdd_example:x1}:
    & [ & $((1, 0), \bot, \top)$ & ]
    \\
    \ref{fig:bdd_example:xor}:
    & [ & $((0, 0), (1,0), (1,1))$ & , & $((1, 0), \bot, \top)$ & , & $((1, 1), \top, \bot)$ & ]
  \end{tabular}
    
  \caption{In-order representation of BDDs of Fig.~\ref{fig:bdd_example}}
  \label{fig:bdd_topological}
\end{figure}

The conventional recursive algorithms traverse the input (and the output) with
random-access as dictated by the call stack. Adiar replaces this stack with a
priority queue that is sorted such that it is synchronised with a sequential
traversal through the input(s). Specifically, the recursion requests
$\arc{s}{}{t}$ from a BDD node $s$ to $t$ is sorted on the target $t$ -- this
way the requests for $t$ are at the top of the priority queue when $t$ is
reached in the input. For example, after processing the root $(0,0)$ of the BDD
in Fig.~\ref{fig:bdd_example:xor}, the priority queue includes the arcs
$\arc{(0,0)}{\bot}{(1,0)}$ and $\arc{(0,0)}{\top}{(1,1)}$, in that order.
Notice, this is exactly in the same order as the sequence of nodes in
Fig.~\ref{fig:bdd_topological}. Essentially, this priority queue maintains the
yet unresolved parts of the recursion tree $(V',A')$ throughout a level by level
top-down sweep. Yet, since the ordering of the priority queue groups together
requests for the same $t$, the graph $(V',A')$ is not a tree but a DAG.

For BDD algorithms that produce an output BDD, e.g.\ the Apply algorithm, Adiar
first constructs $(V',A')$ level by level. When the output BDD node $t \in V'$
is created from nodes $v_f \in V_f$ and $v_g \in V_g$, the top of the priority
queues provides all ingoing arcs, which are placed in the output. Outgoing arcs
to a terminal, $\outset(t) \cap (V' \times \B \times \B)$, are also immediately
placed in a separate output. On the other hand, recursion requests from $t$ to
its yet unresolved non-terminal children, $\outset(t) \setminus (V' \times \B
\times \B)$, have to be processed later. To do so, these unresolved arcs are put
back into the priority queue as arcs
\begin{equation*}
  (\arc{t}{b}{(v_{f}',v_{g}')}) \in V' \times \B \times (V_f \times V_g)
  \enspace ,
\end{equation*}
where the arc's target is the tuple of input nodes $v_f' \in V_f$ and $v_g' \in
V_g$. This essentially makes the priority queue contain all the yet unresolved
arcs of the output. For example, when using Apply to produce
Fig.~\ref{fig:bdd_example:xor} from Fig.~\ref{fig:bdd_example:x0} and
\ref{fig:bdd_example:x1}, the root node of the output is resolved to have uid
$(0,0)$ and the priority queue contains arcs $\arc{(0,0)}{\bot}{(\bot, (1,0))}$
and $\arc{(0,0)}{\top}{(\top, (1,0))}$. Both of these arcs are then later
resolved, creating the nodes $(1,0)$ and $(1,1)$, respectively.

\if\arxiv
\begin{figure}[ht!]
\else
\begin{figure}[t]
\fi
  \centering

  \begin{tikzpicture}[every text node part/.style={align=center}]
    % Boxes
    \draw (0,0) rectangle ++(2,1)
    node[pos=.5]{\texttt{Apply}};
    \draw (5,0) rectangle ++(2,1)
    node[pos=.5]{\texttt{Reduce}};

    % Arcs
    \draw[->] (-0.5,0.9) -- ++(0.5,0)
    node[pos=-1.3]{$f$ \texttt{nodes}};
    \draw[->] (-0.5,0.5) -- ++(0.5,0)
    node[pos=-0.5]{$\odot$};
    \draw[->] (-0.5,0.1) -- ++(0.5,0)
    node[pos=-1.3]{$g$ \texttt{nodes}};

    \draw[->] (2,0.8) -- ++(3,0)
    node[pos=0.5,above]{\small internal \texttt{arcs}};

    \node at (3.5,0.5) {\textcolor{gray}{$f \odot g$ \texttt{arcs}}};

    \draw[->] (2,0.2) -- ++(3,0)
    node[pos=0.5,below]{\small terminal \texttt{arcs}};

    \draw[->] (7,0.5) -- ++(0.5,0)
    node[pos=2.9]{$f \odot g$ \texttt{nodes}};
  \end{tikzpicture}

  \caption{The Apply--Reduce pipeline in \Adiar}
  \label{fig:tandem}
\end{figure}
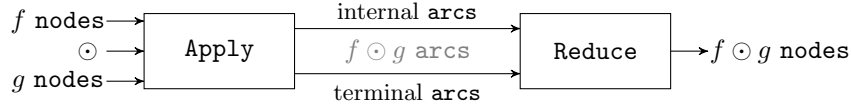

Yet, these \emph{top-down sweeps} of Adiar produce sequences of arcs rather than
nodes. Furthermore, the DAG $(V',A')$ is not necessarily a reduced OBDD. Hence,
as shown in Fig.~\ref{fig:tandem}, Adiar follows up on the above top-down sweep
with a \emph{bottom-up sweep} that I/O-efficiently recreates Bryant's original
Reduce algorithm in \cite{Bryant1986}. Here, a priority queue forwards the uid
of $t'$ that is the result from applying the reduction rules to a BDD node $t$
in $(V',A')$ to the to-be reduced parents $s$ of $t$. These parents are
immediately available by a sequential reading of $(V',A')$ since $\inset(t)$ was
output together within the prior top-down sweep. Both reduction rules are
applied by accumulating all nodes at level $j$ from the arcs in the priority
queue, filtering out \emph{don't care} nodes, sorting the remaining nodes such
that duplicates come in succession and can be eliminated efficiently, and
finally passing the necessary information to their parents via the priority
queue.

\section{Levelised Cuts of a Directed Acyclic Graph} \label{sec:cut}

Any DAG can be divided in one or more ways into several \emph{levels}, where all
vertices at a given level only have outgoing arcs to vertices at later levels.
\begin{definition} \label{def:levelization}
  Given a DAG $(V,A)$ a levelisation of vertices in $V$ is a function $\Lcal : V
  \rightarrow \N \cup \{\infty\}$ such that for any two vertices $v, v' \in V$,
  if there exists an arc $v \rightarrow v'$ in $A$ then $\Lcal(v) < \Lcal(v')$.
\end{definition}
Intuitively, $\Lcal$ is a labeling of vertices $v \in V$ that respects a
topological ordering of $V$. Since $(V,A)$ is a DAG, such a topological ordering
always exists and hence such an $\Lcal$ must also always exist. Specifically,
Specifically, let $\pi_V$ in be the longest path in $(V,A)$ (which must be from
some source $s \in V$ to a sink $t \in V$) and $\pi_v$ be the longest path any
given $v \in V$ to any sink $t \in V$, then $\Lcal(v)$ can be defined to be the
difference of their lengths, i.e.\ $\abs{\pi_V} - \abs{\pi_v}$.

Given a DAG and a levelisation $\Lcal$, we can restrict the freedom of a cut to
be constricted within a small window with respect to $\Lcal$.
Fig.~\ref{fig:i-level_cut_definition} provides a visual depiction of the
following definition.

\begin{definition} \label{def:i-level cut}
  An \emph{$i$-level cut} for $i \geq 1$ is a cut $(S,T)$ of a DAG $(V,A)$ with
  levelisation $\Lcal$ for which there exists a $j \in \N$ such that
  $\Lcal(s) < j + i$ for all $s \in S$ and $\Lcal(t) > j$ for all $t \in T$.
\end{definition}

\begin{figure}[ht!]
  \centering

  \begin{tikzpicture}
    \input{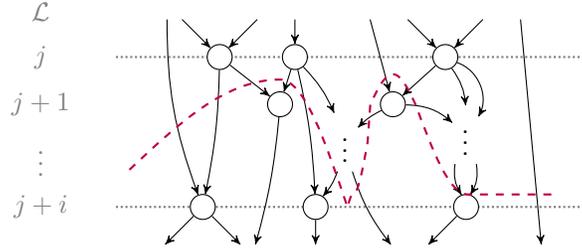}
  \end{tikzpicture}

  \caption{Visualisation of an $i$-level cut.}
  \label{fig:i-level_cut_definition}
\end{figure}

As will become apparent later, deriving the $i$-level cut with maximum weight
for $i \in \{ 1, 2 \}$ will be of special interest.
Fig.~\ref{fig:i-level_cut_examples} shows two $1$-level cuts and three $2$-level
cuts in the BDD for the exclusive-or of the two variables $x_0$ and $x_1$. A
$1$-level cut is by definition a cut between two adjacent levels whereas a
$2$-level cut allows nodes on level $j+1$ to be either in $S$ or in $T$. In
Fig.~\ref{fig:i-level_cut_examples}, both the maximum $1$-level and $2$-level
cuts have size $4$.

\begin{figure}[ht!]
  \centering

  \begin{tikzpicture}
    \input{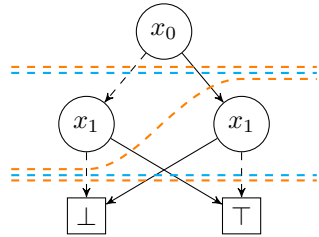}
  \end{tikzpicture}

  \caption{$1$-level (cyan) and $2$-level (orange) cuts in the $x_0 \oplus x_1$
    BDD.}
  \label{fig:i-level_cut_examples}
\end{figure}

\begin{proposition} \label{prop:max 1-level complexity}
  The maximum $1$-level cut in a DAG $(V,A)$ with levelisation $\Lcal$ is
  computable in polynomial time.
\end{proposition}
\begin{proof}
  For a specific $j \in \Lcal(V)$ we can compute the size of the $1$-level cut
  at $j$ in $O(A)$ time by computing the sum of $w((s,t))$ over all arcs $(s, t)
  \in A$ where $\Lcal(s) \leq j$ and $\Lcal(t) > j$. This cut is by definition
  unique for $j$ and hence maximal. Repeating this for each $j \in \Lcal(V)$ we
  obtain the maximum $1$-level cut of the entire DAG in $O(\abs{\Lcal(V)} \cdot
  \abs{A}) = O(\abs{V} \cdot \abs{A})$ time.
\end{proof}

\begin{proposition} \label{prop:max 2-level complexity}
  The maximum $2$-level cut in a DAG $(V,A)$ with levelisation $\Lcal$ is
  computable in polynomial time.
\end{proposition}
\begin{proof}
  Given a level $j \in \Lcal(V)$, any $2$-level cut for $j-1$ has all vertices
  $v \in V$ with $\Lcal(v) \neq j$ fixed to be in $S$ or in $T$. That is, only
  vertices $v$ where $\Lcal(v) = j$ may be part of either $S$ or of $T$. A
  vertex $v$ at level $j$ can greedily be placed in $S$ if $\sum_{a \in
    \outset(v)} w(a) < \sum_{a \in \inset(v)} w(a) $ and in $T$ otherwise. This
  greedy decision procedure runs in $O(\abs{A})$ time for each level, resulting
  in an $O(\abs{\Lcal(V)} \cdot \abs{A}) = O(\abs{V} \cdot \abs{A})$ total
  running time.
\end{proof}

Lampis, Kaouri, and Mitsou~\cite{Lampis2011} prove NP-completeness for computing
the maximum cut of a DAG by a reduction from the \emph{not-all-equal SAT
  problem} (\textsc{nae3sat}) to a DAG with $5$ levels. That is, they prove
NP-completeness for computing the size of the maximum $i$-level cut for $i \geq
4$. This still leaves the complexity of the maximum $i$-level cut for $i = 3$ as
an open problem.

\subsection{Maximum Levelised Cuts in BDD Manipulation} \label{sec:cut:bdd}

For an OBDD, represented by the DAG $(V,A)$, we will consider the levelisation
function $\LcalBDD$ where all nodes with the same label are on the same level.
\begin{equation*}
  \LcalBDD(v) \triangleq
  \begin{cases}
    \labelof(v) & \text{if } v \not\in \B
    \\
    \infty      & \text{if } v \in \B
  \end{cases}
\end{equation*}

For a BDD $f$ with the DAG $(V,A)$, let $N_f \triangleq \abs{V \setminus \B}$ be
the number of internal nodes in $V$. Let $C_{i:f}$ denote the size of the
unweighted maximum $i$-level cut in $(V,A)$; in Section~\ref{sec:cut:terminal}
we will consider weighted maximum cuts, where one or more terminals are ignored.
Finally, we introduce the arc $\arc{(- \infty)}{}{r_f}$ to the root. This
simplifies the results that follow since $\inset(v) \neq \emptyset$ for all $v
\in V$.

\begin{lemma}
  The maximum cut of a multi-rooted decision diagram $(V,A)$ is less than or
  equals to $N + r$ where $N = \abs{V \setminus \B}$ is the number of internal
  nodes and $r \geq 1$ is the number of roots.
\end{lemma}
\begin{proof}
  We will prove this by induction on the number of internal nodes, $N$.

  For $N=1$, the decision diagram must be a singly rooted DAG with a single node
  $v$ with two outgoing arcs to $\B$, e.g.\ a BDD for the function $x_i$. The
  largest cut is of size $2$ which equals the desired bound.

  Assume for $N'$ that any decision diagrams with $N'$ number of internal nodes
  and some $r'$ number of roots have a maximum cut with a cost of at most $N' +
  r'$. Consider a decision diagram $(V,A)$ with $N=N'+1$ internal nodes and $r
  \geq 1$ number of roots. Let $v$ be one of the $r$ roots. After removing $v$,
  the resulting decision diagram $(V',A')$ has $r' = r +
  \delta_{\abs{\inset(v.\mathit{low})} = 1} +
  \delta_{\abs{\inset(v.\mathit{high})} = 1} - 1$ roots where $\delta$ is the
  indicator function. The number of internal nodes in $(V',A')$ is $N'$ and so
  the maximum size of its cut is by induction $N' + r'$.

  We will now argue, that adding $v$ back into the DAG $(V',A')$ may not
  increase the cut by more than one. Notice, since each node in a decision
  diagram is binary, we may assume that ingoing arcs to a node $v'$ are only
  contributing to a cut if $\abs{\inset(v')} > 2$. Hence, the arc
  $\arc{v}{\bot}{v.\mathit{low}}$ may only contribute to the maximum cut in
  $(V,A)$, if $\abs{\inset(v.\mathit{low})} > 2$. By definition, this means
  $\arc{v}{\bot}{v.\mathit{low}}$ may only contribute to the maximum cut, if
  $\delta_{\abs{\inset(v.\mathit{low})}} = 0$. Symmetrically,
  $\delta_{\abs{\inset(v.\mathit{high})}}$ accounts for whether this very arc
  may be removed from the cut or not. Since
  $\delta_{\abs{\inset(v.\mathit{low})} = 1} +
  \delta_{\abs{\inset(v.\mathit{high})} = 1} \leq 2$, we have $r' \leq r + 1$.
  That is, adding the two arcs of $v$ into $(V',A')$ may only add one arc to the
  maximum cut that is not associated with a root of the DAG and so $N' + r' \leq
  N' + r + 1 = N + r$ is an upper bound on the maximum cut of $(V,A)$ as
  desired.
\end{proof}

By applying this to a single BDD, we obtain the following simple upper bound on
any maximum cut of its DAG.

\begin{theorem} \label{prop:max cut bound}
  The maximum cut of the BDD $f$ has a size of at most $N_f + 1$.
\end{theorem}

This bound is tight for $i$-level cuts, as is evident from
Fig.~\ref{fig:i-level_cut_examples} where the size of the maximum ($i$-level)
cut is $4$. Yet, in general, one can obtain a better upper bound on the maximum
$i$-level cut of the (unreduced) output of each BDD operation when the maximum
$i$-level cut of the input is known.

\begin{theorem} \label{prop:apply 2-level cut} \label{prop:apply 1-level cut}
  For $i \in \{ 1,2 \}$, the maximum $i$-level cut of the (unreduced) output of
  Apply of $f$ and $g$ is at most $C_{i:f} \cdot C_{i:g}$.
\end{theorem}
\begin{proof}
  Let us only consider the more complex case of $i = 2$; the proof for $i = 1$
  follows from the same line of thought.

  Every node of the output represents a tuple $(v_f, v_g)$ where $v_f$, resp.\
  $v_g$, is an internal node of $f$, resp.\ $g$, or is one of the terminals $\B
  = \{ \bot, \top \}$. An example of this situation is shown in
  Fig.~\ref{fig:apply_cut}. The node $(v_f, v_g)$ contributes with
  $\max(\abs{\inset((v_f, v_g))}, \abs{\outset((v_f, v_g))})$ to the maximum
  $2$-level cut at that level. Since it is a BDD node, $\abs{\outset((v_f,v_g))}
  = 2$. We have that $\abs{\inset((v_f, v_g))} \leq \abs{\inset(v_f)} \cdot
  \abs{\inset(v_g)}$ since all combinations of in-going arcs may potentially
  exist and lead to this product of $v_f$ and $v_g$. Expanding on this, we
  obtain
  \begin{align*}
    \abs{\inset((v_f, v_g))}
    &\leq
    \abs{\inset(v_f)} \cdot \abs{\inset(v_g)}
    \\ &\leq
    \max(\abs{\inset(v_f)}, \abs{\outset(v_f)})
      \cdot \max(\abs{\inset(v_g)}, \abs{\outset(v_g)})
    \enspace .
  \end{align*}
  That is, the maximum $2$-level cut for a level is less than or equal to the
  product of the maximum $2$-level cuts of the input at the same level. Taking
  the maximum $2$-level cut across all levels we obtain the final product of
  $C_{2:f}$ and $C_{2:g}$.
\end{proof}

\if\arxiv
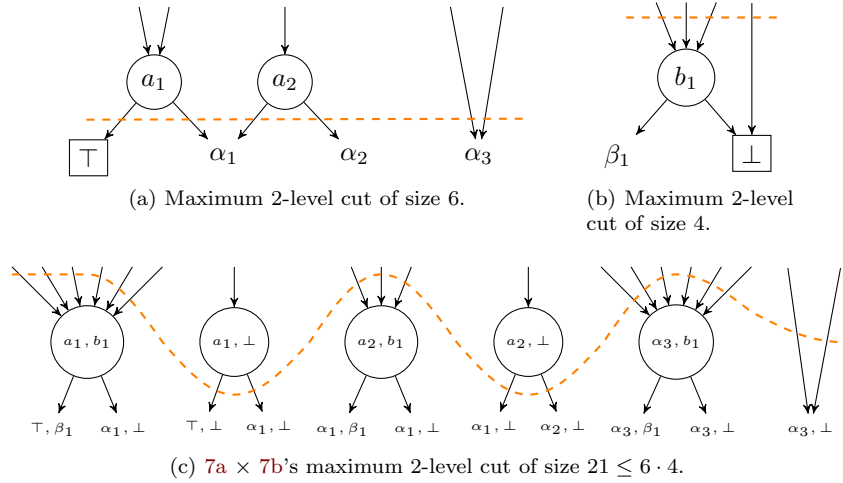
\begin{figure}[ht!]
\else
\begin{figure}[tb]
\fi
  \centering

  \subfloat[Maximum $2$-level cut of size $6$.]{
    \label{fig:apply_cut:a}
    \centering

    \begin{tikzpicture}
      % nodes of level j
      \node[shape = circle, draw = black] (a1) {$a_1$};
      \node[shape = circle, draw = black, right=of a1] (a2) {$a_2$};

      % subtrees of levels > j
      \node[shape = rectangle, draw = black, below left=0.5cm and .35cm of a1] (top) {$\top$};
      \node[below right=0.5cm and .35cm of a1] (alpha1) {$\alpha_1$};
      \node[below right=0.5cm and .35cm of a2] (alpha2) {$\alpha_2$};
      \node[below right=0.5cm and 2.0cm of a2] (alpha3) {$\alpha_3$};

      % arcs in the max-cut
      \draw[->]
        ($ (a1) + (-0.2, 1) $) edge (a1)
        ($ (a1) + (0.2, 1) $)  edge (a1)
        (a1)                   edge (alpha1)
        (a1)                   edge (top)
        ($ (a2) + (0, 1) $)    edge (a2)
        (a2)                   edge (alpha1)
        (a2)                   edge (alpha2)
        ($ (a2) + (2.2, 1) $)  edge (alpha3)
        ($ (a2) + (2.9, 1) $)   ->  (alpha3)
      ;

      % cut line
      \draw[thick, dashed, 2level]
        ($ (a1) - (0.9, 0.5) $) edge ($ (alpha3) + (0.6, 0.5) $)
      ;
    \end{tikzpicture}
  }
  \qquad
  \subfloat[Maximum $2$-level cut of size $4$.]{
    \label{fig:apply_cut:b}
    \centering

    \begin{tikzpicture}
      % nodes of level j
      \node[shape = circle, draw = black] (b1) {$b_1$};

      % subtrees of levels > j
      \node[below left=0.5cm and .35cm of a1]  (beta1) {$\beta_1$};
      \node[shape = rectangle, draw = black, below right=0.5cm and .35cm of a1] (bot) {$\bot$};

      % arcs in the max-cut
      \draw[->]
        ($ (b1) + (-0.4, 1) $) edge (b1)
        ($ (b1) + (0, 1) $)    edge (b1)
        ($ (b1) + (0.4, 1) $)  edge (b1)
        (b1) edge (beta1)
        (b1) edge (bot)
        ($ (bot) + (0, 2) $)    ->  (bot)
      ;

      % cut line
      \draw[thick, dashed, 2level]
        ($ (b1) + (-0.8, 0.8) $) edge ($ (b1) + (1.3, 0.8) $)
      ;
    \end{tikzpicture}
  }

  \subfloat[\ref{fig:apply_cut:a} $\times$ \ref{fig:apply_cut:b}'s maximum $2$-level cut of size $21 \leq 6 \cdot 4$.]{
    \label{fig:apply_cut:prod}
    \centering

    \begin{tikzpicture}
      % nodes
      \node[shape = circle, draw = black]                    (a1_b1)     {\tiny $a_1, b_1$};
      \node[below left=0.6cm and -0.3cm of a1_b1] (a1_b1__1) {\tiny $\top, \beta_1$};
      \node[below right=0.6cm and -0.3cm of a1_b1] (a1_b1__2) {\tiny $\alpha_1, \bot$};

      \node[shape = circle, draw = black, right=of a1_b1]    (a1_bot)  {\tiny $a_1, \bot$};
      \node[below left=0.6cm and -0.3cm of a1_bot] (a1_bot__1) {\tiny $\top, \bot$};
      \node[below right=0.6cm and -0.3cm of a1_bot] (a1_bot__2) {\tiny $\alpha_1, \bot$};

      \node[shape = circle, draw = black, right=of a1_bot] (a2_b1)     {\tiny $a_2, b_1$};
      \node[below left=0.6cm and -0.3cm of a2_b1]  (a2_b1__1) {\tiny $\alpha_1, \beta_1$};
      \node[below right=0.6cm and -0.3cm of a2_b1] (a2_b1__2) {\tiny $\alpha_1, \bot$};
      
      \node[shape = circle, draw = black, right=of a2_b1]    (a2_bot)  {\tiny $a_2, \bot$};
      \node[below left=0.6cm and -0.3cm of a2_bot]  (a2_bot__1) {\tiny $\alpha_1, \bot$};
      \node[below right=0.6cm and -0.3cm of a2_bot] (a2_bot__2) {\tiny $\alpha_2, \bot$};

      \node[shape = circle, draw = black, right=of a2_bot] (alpha3_b1) {\tiny $\alpha_3, b_1$};
      \node[below left=0.6cm and -0.3cm of alpha3_b1]  (alpha3_b1__1) {\tiny $\alpha_3, \beta_1$};
      \node[below right=0.6cm and -0.3cm of alpha3_b1] (alpha3_b1__2) {\tiny $\alpha_3, \bot$};

      \node[below right=0.6cm and 1cm of alpha3_b1] (alpha3_bot) {\tiny $\alpha_3, \bot$};

      % arcs in the max-cut
      \draw[->]
        ($ (a1_b1) + (-1, 1) $)           edge (a1_b1)
        ($ (a1_b1) + (-0.6, 1) $)         edge (a1_b1)
        ($ (a1_b1) + (-0.2, 1) $)         edge (a1_b1)
        ($ (a1_b1) + (0.2, 1) $)          edge (a1_b1)
        ($ (a1_b1) + (0.6, 1) $)          edge (a1_b1)
        ($ (a1_b1) + (1, 1) $)            edge (a1_b1)
        (a1_b1)                           edge (a1_b1__1)
        (a1_b1)                           edge (a1_b1__2)
        ($ (a1_bot) + (0, 1) $)           edge (a1_bot)
        (a1_bot)                          edge (a1_bot__1)
        (a1_bot)                          edge (a1_bot__2)
        ($ (a2_b1) + (-0.4, 1) $)         edge (a2_b1)
        ($ (a2_b1) + (0, 1) $)            edge (a2_b1)
        ($ (a2_b1) + (0.4, 1) $)          edge (a2_b1)
        (a2_b1)                           edge (a2_b1__1)
        (a2_b1)                           edge (a2_b1__2)
        ($ (a2_bot) + (0, 1) $)           edge (a2_bot)
        (a2_bot)                          edge (a2_bot__1)
        (a2_bot)                          edge (a2_bot__2)
        ($ (alpha3_b1) + (-1, 1) $)       edge (alpha3_b1)
        ($ (alpha3_b1) + (-0.6, 1) $)     edge (alpha3_b1)
        ($ (alpha3_b1) + (-0.2, 1) $)     edge (alpha3_b1)
        ($ (alpha3_b1) + (0.2, 1) $)      edge (alpha3_b1)
        ($ (alpha3_b1) + (0.6, 1) $)      edge (alpha3_b1)
        ($ (alpha3_b1) + (1, 1) $)        edge (alpha3_b1)
        (alpha3_b1)                       edge (alpha3_b1__1)
        (alpha3_b1)                       edge (alpha3_b1__2)
        ($ (alpha3_bot) + (-0.3, 2.15) $) edge (alpha3_bot)
        ($ (alpha3_bot) + (0.4, 2.15) $)   -> (alpha3_bot)
      ;
      
      % % cut line
      \draw[thick, dashed, 2level]
        ($ (a1_b1) + (-1.0, 0.9) $) --
        ($ (a1_b1) + (0.0, 0.9) $) cos
        ($ (a1_b1) + (1, 0) $) sin
        ($ (a1_bot) + (0, -0.7) $) cos
        ($ (a1_bot) + (1, 0) $) sin
        ($ (a2_b1) + (0, 0.9) $) cos
        ($ (a2_b1) + (1, 0) $) sin
        ($ (a2_bot) + (0, -0.7) $) cos
        ($ (a2_bot) + (1, 0) $) sin
        ($ (alpha3_b1) + (0, 0.9) $) cos
        ($ (alpha3_b1) + (1, 0.4) $) sin
        ($ (alpha3_b1) + (2.2, 0) $)
      ;
    \end{tikzpicture}
  }
  
  \caption{Relation between the maximal $2$-level cut of two BDDs' internal arcs
    and the maximum $2$-level cut of their product.}
  \label{fig:apply_cut}
\end{figure}

The bounds in Thm.~\ref{prop:apply 2-level cut} are better than what can be
derived from Thm.~\ref{prop:max cut bound} since $C_{i:f}$ and $C_{i:g}$ are
themselves cuts and hence their product must be at most the bound based on the
possible number of nodes. They are also tight: the maximum $i$-level cut for $i
\in \{1,2\}$ of the BDDs for the variables $x_0$ and $x_1$ in
Fig~\ref{fig:bdd_example:x0} and \ref{fig:bdd_example:x1} both have size $2$
while the BDD for the exclusive-or of them in Fig.~\ref{fig:bdd_example:xor}
has, as shown in Fig.~\ref{fig:i-level_cut_examples}, a maximum $i$-level cut of
size $4$.

Since the maximum $1$-level cut also bounds the number of outgoing arcs of all
nodes on each level, one can derive an upper bound on the output's width. That
is, based on Thm.~\ref{prop:apply 1-level cut} we can obtain the following
interesting result.
\begin{corollary} \label{prop:apply width}
  The \emph{width} of Apply's output is less than or equal to $\tfrac{1}{2}
  \cdot C_{1:f} \cdot C_{1:g}$.
\end{corollary}
\begin{proof}
  The $\tfrac{1}{2}$ compensates for the outdegree of each BDD node.
\end{proof}

This is only an upper bound, as half of the arcs that cross the widest level are
also counted. Yet, it is tight, as Fig.~\ref{fig:i-level_cut_examples} has a
maximum $i$-level cut of size $4$ and a width of $2$.

Thm.~\ref{prop:apply 2-level cut} is of course only an over-approximation. The
gap between the upper bound and the actual maximum $i$-level cut arises because
Thm.~\ref{prop:apply 2-level cut} does not account for pairs $(v_f,v_g)$, where
node $v_f$ sits above $f$'s maximum $2$-level cut and $v_g$ sits below $g$'s
maximum $2$-level cut, and vice versa. In this case, outgoing arcs of $v_f$ are
paired with ingoing arcs of $v_g$, even though this would be strictly larger
than the arcs of their product. Furthermore, similar to Thm.~\ref{prop:max cut
  bound}, this bound does not account for arcs that cannot be paired as they
reflect conflicting assignments to one or more input variables. For example, in
the case where the out-degree is greater for both nodes, the above bound
mistakenly pairs the low arcs with the high arcs and vice versa.

\subsection{Improving Bounds by Accounting for Terminal Arcs} \label{sec:cut:terminal}

Some of the imprecision in the over-approximation of Thm.~\ref{prop:apply
  2-level cut} highlighted above can partially be addressed by explicitly
accounting for the arcs to each terminal. For $B \subseteq \B$, let $w_B$ be the
weight function that only cares for arcs to internal BDD nodes and to the
terminals in $B$.
\begin{equation*}
  w_B(\arc{s}{b}{t}) =
  \begin{cases}
    1 &\text{if } t \in V \setminus \B \text{ or }  t \in B
    \\
    0 &\text{otherwise}
  \end{cases}
  \enspace .
\end{equation*}
Let $C_{i:f}^B$ be the maximum $i$-level cut of $f$ with respect to $\LcalBDD$
and $w_B$.

The constant hidden within the $O(\abs{V}\cdot\abs{A})$ running time of the
algorithm in the proof of Prop.~\ref{prop:max 1-level complexity} is smaller
than the one in the proof of Prop.~\ref{prop:max 2-level complexity}. Hence, the
following slight over-approximations of $C_{2:f}^B$ given $C_{1:f}^B$ may be
useful.

\begin{lemma} \label{prop:2-level from 1-level cut:1}
  The maximum $2$-level cut $C_{2:f}^\emptyset$ is less than or equals to
  $\tfrac{3}{2} \cdot C_{1:f}^\emptyset$.
\end{lemma}
\begin{proof}
  $C_{1:f}^\emptyset$ is an upper bound on the number of ingoing arcs to nodes
  on level $j+1$ for any $j$. This places the BDD nodes $v$ with $\Lcal(v) =
  j+1$ in the $S$ partition of the $1$-level cut. The only case where such a $v$
  should be moved to the $S$ partition for the maximum $2$-level cut at level
  $j$ is if $\abs{\inset(v)} = 1$ and $\abs{\outset(v)} = 2$ in the subgraph
  only consisting of internal arcs. Since $C_{1:f}^\emptyset$ is also an upper
  bound on the number of outgoing arcs then at most $C_{1:f}^\emptyset / 2$
  nodes at level $j$ may be moved to $S$ to then count their $C_{1:f}^\emptyset$
  outgoing arcs. This leaves $C_{1:f}^\emptyset / 2$ ingoing arcs still to be
  counted. Combining both, we obtain the bound above.
\end{proof}

\begin{lemma} \label{prop:2-level from 1-level cut:2}
  For $B \subseteq \B$, $C_{2:f}$ is at most $\tfrac{1}{2} \cdot
  C_{1:f}^\emptyset + C_{1:f}^B$.
\end{lemma}
\begin{proof}
  The $C_{1:f}^B - C_{1:f}^\emptyset$ is the number of arcs to terminals. The
  remaining $C_{1:f}^\emptyset$ may be arcs to a BDD node where up to
  $C_{1:f}^\emptyset / 2$ can, as in Lem.~\ref{prop:2-level from 1-level cut:1},
  be moved to the other side of the cut to increase the $2$-level cut with
  $C_{1:f}^\emptyset / 2$. Simplifying $\tfrac{3}{2} + (C_{1:f}^B -
  C_{1:f}^\emptyset)$ we obtain the desired bound.
\end{proof}

Finally, we can tighten the bound in Thm.~\ref{prop:apply 2-level cut} by making
sure (1) not to unnecessarily pair terminals in $f$ with terminals in $g$ and
(2) not to pair terminals from $f$ and $g$ with nodes of the other when said
terminal shortcuts the operator.
\begin{lemma}
  The maximum $2$-level cut of the (unreduced) output $f \odot g$ of Apply
  excluding arcs to terminals, $C_{2:f \odot g}^\emptyset$, is at most
  \begin{equation*}
      C_{2:f}^{B_{\mathit{left}(\odot)}} \cdot C_{2:g}^\emptyset
    + C_{2:f}^\emptyset \cdot C_{2:g}^{B_{\mathit{right}(\odot)}}
    - C_{2:f}^\emptyset \cdot C_{2:g}^\emptyset
    % C_{2:f}^\emptyset \cdot C_{2:g}^\emptyset
    % + (C_{2:f}^{B_{\mathit{left}(\odot)}} - C_{2:f}^\emptyset) \cdot C_{2:g}^\emptyset
    % + C_{2:f}^\emptyset \cdot (C_{2:g}^{B_{\mathit{right}(\odot)}} - C_{2:g}^\emptyset)
    \enspace ,
  \end{equation*}
  where $B_{\mathit{left}(\odot)}, B_{\mathit{right}(\odot)} \subseteq \B$ are
  the terminals that do not shortcut $\odot$.
\end{lemma}

\subsection{Maximum Levelised Cuts in ZDD Manipulation} \label{sec:cut:zdd}

The results in Section~\ref{sec:cut:bdd} and \ref{sec:cut:terminal} are loosely
yet subtly coupled to the reduction rules of BDDs. Specifically,
Thm.~\ref{prop:max cut bound} is applicable to ZDDs as-is but
Thm.~\ref{prop:apply 2-level cut} and its derivatives provide unsound bounds for
ZDDs. This is due to the fact that, unlike for BDDs, a suppressed ZDD node may
re-emerge during a ZDD product construction algorithm. For example in the case
of the \emph{union} operation, when processing a pair of nodes with two
different levels, its high child becomes the product of a node $v$ in one ZDD
and the $\bot$ terminal in the other -- even if there was no arc to $\bot$ in
the original two cuts for $f$ and $g$.

The solution is to introduce another special arc similar to
$\arc{(-\infty)}{}{r_f}$ which accounts for this specific case: if there are no
arcs to $\bot$ to pair with, then the arc $\arc{(-\infty)}{}{\bot}$ is counted
as part of the input's cut. That is, all prior results for BDDs apply to ZDDs,
assuming $C_{i:f}^B$ is replaced with $\mathit{ZC}_{i:f}^B$ defined to be
\begin{equation*}
  \mathit{ZC}_{i:f}^B =
  \begin{cases}
    C_{i:f}^B + 1 & \text{if } \bot \in B \text{ and } C_{i:f}^B = C_{i:f}^{B \setminus \{\bot\}}
    \\
    C_{i:f}^B     & \text{otherwise}
  \end{cases}
  \enspace .
\end{equation*}

\subsection{Adding Levelised Cuts to Adiar's Algorithms} \label{sec:cut:adiar}

The description of Adiar in Section~\ref{sec:preliminaries:adiar} leads to the
following observations.
\begin{itemize}
\item The contents of the priority queues in the top-down Apply algorithms are
  always a $1$-level or a $2$-level cut of the input or of the output --
  possibly excluding arcs to one or both terminals.

\item The contents of the priority queue in the bottom-up Reduce algorithm are
  always a $1$-level cut of the input, excluding any arcs to terminals.
\end{itemize}
Specifically, the priority queues always contain an $i$-level cut $(S,T)$, where
$S$ is the set of processed diagram nodes and $T$ is the set of yet unresolved
diagram nodes. For example, the $2$-level cuts depicted in
Fig.~\ref{fig:i-level_cut_examples} reflect the states of the top-down priority
queue within the Apply to compute the exclusive-or of
Fig.~\ref{fig:bdd_example:x0} and \ref{fig:bdd_example:x1} to create
Fig.~\ref{fig:bdd_example:xor}. In turn, the $1$-level cuts in
Fig.~\ref{fig:i-level_cut_examples} are also the state of the bottom-up priority
queue of the Reduce sweep that follows.

\if\arxiv1
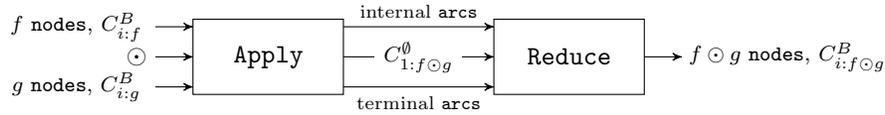
\begin{figure}[ht!]
\else
\begin{figure}[b]
\fi
  \centering

  \begin{tikzpicture}[every text node part/.style={align=center}]
    % Boxes
    \draw (0,0) rectangle ++(2,1)
    node[pos=.5]{\texttt{Apply}};
    \draw (4,0) rectangle ++(2,1)
    node[pos=.5]{\texttt{Reduce}};

    % Arcs
    \draw[->] (-0.5,0.9) -- ++(0.5,0)
    node[pos=-2.1]{\footnotesize $f$ \texttt{nodes}, $C_{i:f}^B$};

    \draw[->] (-0.5,0.5) -- ++(0.5,0)
    node[pos=-0.5]{$\odot$};

    \draw[->] (-0.5,0.1) -- ++(0.5,0)
    node[pos=-2.1]{\footnotesize $g$ \texttt{nodes}, $C_{i:g}^B$};

    \draw[->] (2,0.9) -- ++(2,0)
    node[pos=0.5,above]{\scriptsize internal \texttt{arcs}};

    \draw[->] (2,0.5) -- ++(2,0)
    node[pos=0.5, fill=white]{\footnotesize $C_{1:f \odot g}^\emptyset$};

    \draw[->] (2,0.1) -- ++(2,0)
    node[pos=0.5,below]{\scriptsize terminal \texttt{arcs}};

    \draw[->] (6,0.5) -- ++(0.5,0)
    node[pos=3.8]{\footnotesize $f \odot g$ \texttt{nodes}, $C_{i:f \odot g}^B$};
  \end{tikzpicture}

  \caption{The Apply--Reduce pipeline in \Adiar\ with $i$-level cuts.
    \if\arxiv
    $C_{i:f}^B$ is any $i$-level cut for $i \in \{1,2\}$ and $B \subseteq \B$.\fi}
  \label{fig:tandem with cuts}
\end{figure}

Hence, the upper bounds on the $1$ and $2$-level cuts in
Section~\ref{sec:cut:bdd}, \ref{sec:cut:terminal}, and \ref{sec:cut:zdd} are
also upper bounds on the size of all auxiliary data structures. That is, upper
bounds on the $i$-level cuts of the input can be used to derive a sound
guarantee of whether the much faster internal memory variants can fit into
memory. To only add a minimal overhead to the performance, computing these
$i$-level cuts should be done as part of the preceding algorithm that created
the very input. This extends the tandem in Fig.~\ref{fig:tandem} as depicted in
Fig.~\ref{fig:tandem with cuts} with the $i$-level cuts necessary for the next
algorithm.

What is left is to compute within each sweep an upper bound on these cuts.

\subsubsection{$1$-Level Cut within Top-down Sweeps} \label{sec:cut:adiar:apply}

The priority queues of a top-down sweep only contain arcs between non-terminal
nodes of its output. While their contents in general form a $2$-level cut, the
sweep also enumerates all $1$-level cuts when it has finished processing one
level, and is about to start processing the next. That is, the top-down
algorithm that constructs the unreduced decision diagram $(V',A')$ for $f'$ can
compute $C_{1:f'}^\emptyset$ in $O(\abs{\LcalBDD(V')})$ time by accumulating the
maximum size of its own priority queue when switching from one level to another.
The number of I/O operations is not affected at all.

\subsubsection{$i$-Level Cuts within the Bottom-up Reduce} \label{sec:cut:adiar:reduce}

To compute the $1$-level and $2$-level cuts of the output during the Reduce
algorithm, the algorithms in the proofs of Prop.~\ref{prop:max 1-level
  complexity} and \ref{prop:max 2-level complexity} need to be incorporated.
Since the Reduce algorithm works bottom-up, it cannot compute these cuts
exactly: the bottom-up nature only allows information to flow from lower levels
upwards while an exact result also requires information to be passed downwards.
Specifically, Fig.~\ref{fig:red1_chain_case} shows an unreduced BDD whose
maximum $1$ and $2$-level cut is increased due to the reduction removing nodes
above the cut. Both over-approximation algorithms below are tight since for the
input in Fig.~\ref{fig:red1_chain_case} they compute the exact result.

\if\arxiv
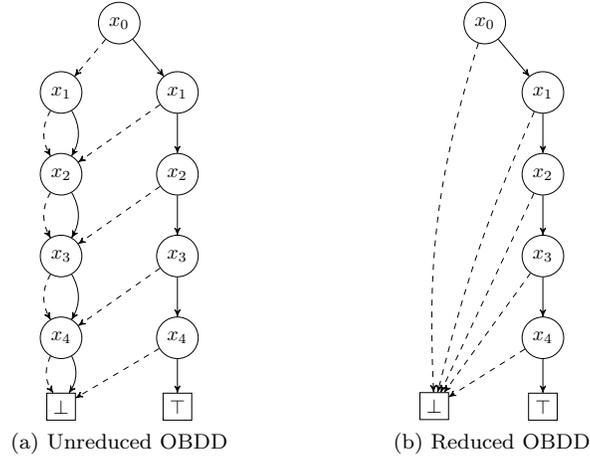
\begin{figure}[t]
\else
\begin{figure}[t]
\fi
  \centering

  \subfloat[Unreduced OBDD]{
    \hspace{1cm}\scalebox{0.75}{\begin{tikzpicture}
      % nodes
      \node[shape = circle, draw = black]                                  (0)   {$x_0$};
      \node[shape = circle, draw = black, below left=0.7cm and .5cm of 0]  (1_0) {$x_1$};
      \node[shape = circle, draw = black, below right=0.7cm and .5cm of 0] (1_1) {$x_1$};
      \node[shape = circle, draw = black, below=0.7cm of 1_0]              (2_0) {$x_2$};
      \node[shape = circle, draw = black, below=0.7cm of 1_1]              (2_1) {$x_2$};
      \node[shape = circle, draw = black, below=0.7cm of 2_0]              (3_0) {$x_3$};
      \node[shape = circle, draw = black, below=0.7cm of 2_1]              (3_1) {$x_3$};
      \node[shape = circle, draw = black, below=0.7cm of 3_0]              (4_0) {$x_4$};
      \node[shape = circle, draw = black, below=0.7cm of 3_1]              (4_1) {$x_4$};

      % leafs
      \node[shape = rectangle, draw = black, below=.6cm of 4_0]  (sink_F) {$\bot$};
      \node[shape = rectangle, draw = black, below=.6cm of 4_1]  (sink_T) {$\top$};

      % arcs
      \draw[->]
      (0) edge (1_1)
      (1_0) edge[bend left] (2_0)
      (1_1) edge (2_1)
      (2_0) edge[bend left] (3_0)
      (2_1) edge (3_1)
      (3_0) edge[bend left] (4_0)
      (3_1) edge (4_1)
      (4_0) edge[bend left] (sink_F)
      (4_1) edge (sink_T)
      ;

      \draw[->, dashed]
      (0) edge (1_0)
      (1_0) edge[bend right]  (2_0)
      (1_1) edge  (2_0)
      (2_0) edge[bend right]  (3_0)
      (2_1) edge  (3_0)
      (3_0) edge[bend right]  (4_0)
      (3_1) edge  (4_0)
      (4_0) edge[bend right]  (sink_F)
      (4_1) edge  (sink_F)
      ;
    \end{tikzpicture}}\hspace{1cm}
  }
  \qquad
  \subfloat[Reduced OBDD]{
    \hspace{1cm}\scalebox{0.75}{\begin{tikzpicture}
      % nodes
      \node[shape = circle, draw = black]                                  (0)   {$x_0$};
      \node[shape = circle, draw = black, below right=0.7cm and .5cm of 0] (1_1) {$x_1$};
      \node[shape = circle, draw = black, below=0.7cm of 1_1]              (2_1) {$x_2$};
      \node[shape = circle, draw = black, below=0.7cm of 2_1]              (3_1) {$x_3$};
      \node[shape = circle, draw = black, below=0.7cm of 3_1]              (4_1) {$x_4$};

      % leafs
      \node[shape = rectangle, draw = black, below left=.7cm and 1.4cm of 4_1]  (sink_F) {$\bot$};
      \node[shape = rectangle, draw = black, below=.6cm of 4_1]  (sink_T) {$\top$};

      % arcs
      \draw[->]
      (0) edge (1_1)
      (1_1) edge (2_1)
      (2_1) edge (3_1)
      (3_1) edge (4_1)
      (4_1) edge (sink_T)
      ;

      \draw[->, dashed]
      (0) edge[bend right=10] (sink_F)
      (1_1) edge[bend right=5] (sink_F)
      (2_1) edge (sink_F)
      (3_1) edge (sink_F)
      (4_1) edge (sink_F)
      ;
    \end{tikzpicture}}\hspace{1cm}
  }
  
  \caption{Example of reduction increasing the $1$ and $2$-level maximum cut.}
  \label{fig:red1_chain_case}
\end{figure}

\paragraph{Over-approximating the $1$-level Cut.}

Starting from the bottom, when processing a level $k \in \LcalBDD(V)$ we may
over-approximate the $1$-level cut $C_{1:f}^B$ for $B \subseteq \{ \bot, \top
\}$ at $j = k$ by summing the following four disjoint contributions.
\begin{enumerate}
\item \label{reduce:1-level:pq content}
  After having obtained all outgoing arcs for unreduced nodes for level $k$, the
  priority queue only contains outgoing arcs from a level $\ell < k$ to a level
  $\ell' > k$. All of these arcs (may) contribute to the cut.

\item \label{reduce:1-level:unread terminals}
  After having obtained all outgoing arcs for level $k$, all yet unread arcs to
  terminals $b \in B$ are from some level $\ell < k$ and (may) contribute to the
  cut.

\item  \label{reduce:1-level:rule1}
  BDD nodes $v$ removed by the first reduction rule in favor of its reduced
  child $v'$ and $w_B(\arc{\_}{}{v'}) = 1$ (may) contribute up to
  $\abs{\inset(v')}$ arcs to the cut.

\item  \label{reduce:1-level:rule2}
  BDD nodes $v'$ that are output on level $k$ after merging duplicates
  (definitely) contribute with $w_B(v'.\mathit{low}) + w_B(v'.\mathit{high})$
  arcs to the cut.
\end{enumerate}
\ref{reduce:1-level:pq content} and~\ref{reduce:1-level:unread terminals} can be
obtained with some bookkeeping on the priority queue and the contents of the
file containing arcs to terminals. \ref{reduce:1-level:rule2} can be resolved
when reduced nodes are pushed to the output. Yet, \ref{reduce:1-level:rule1}
cannot just use the immediate indegree of the removed node $v$ since, as in
Fig.~\ref{fig:red1_chain_case}, it may be part of a longer chain of redundant
nodes. Here, the actual contribution to the cut at level $j = k$ is the indegree
to the entire chain ending in $v$. Due to the single bottom-up sweep style of
the Reduce algorithm, the best we can do is to assume the worst and always count
reduced arcs $\arc{s'}{}{t'}$ where a node $v$ has been removed between $s'$ and
$t'$ as part of the maximum cut.

\paragraph{Over-approximating the $2$-level Cut.}

The above over-approximation of the $1$-level cut can be extended to recreate
the greedy algorithm from the proof of Prop.~\ref{prop:max 2-level complexity}.
Notice, the $1$-level $(S,T)$ cut mentioned before places all nodes of level $j$
in $S$, whereas these nodes are free to be moved to $T$ in the $2$-level cut for
$j-1$. Specifically, Part~\ref{reduce:1-level:rule2} should be changed such that
$v'$ contributes with
\begin{equation*}
  \max(w_B(v'.\mathit{low}) + w_B(v'.\mathit{high}), \abs{\inset(v')})
  \enspace .
\end{equation*}

This requires knowing $\abs{\inset(v')}$. The Reduce algorithm in
\cite{Soelvsten2022:TACAS} reads from a file containing the parents of an
unreduced node $v$, so information about the reduced result $v'$ can be
forwarded to its unreduced parents. Hence, one can accumulate the number of
parents, $\abs{\inset(v)}$. If $\abs{\inset(v')}$ is not affected by the first
reduction rule then this is an upper bound of $\abs{\inset(v')}$. Otherwise, it
still is sound in combination with the above over-counting to solve the
\ref{reduce:1-level:rule1}$^{\text{rd}}$ type of contribution.

\section{Experimental Evaluation} \label{sec:experiments}

We have extended Adiar to incorporate the ideas presented in
Section~\ref{sec:cut} to address the issues highlighted in
Section~\ref{sec:introduction}. Each algorithm has been extended to compute
sound upper bounds for the next phase. Based on these, each algorithm chooses
during initialisation between running with \TPIE's internal or external memory
data structures. This choice is encapsulated within C++ templates, which avoids
introducing any costly indirection when using the auxiliary data structures
since in both cases their type is already known to the compiler.

Section~\ref{sec:cut:adiar} motivates the following three levels of
granularity:
\begin{itemize}
\item \textbf{\#nodes:} Thm.~\ref{prop:max cut bound} is used based on knowing
  the number of internal nodes in the input and deriving the trivial worst-case
  size of the output.

\item \textbf{$1$-level:} Extends \#nodes with Thm.~\ref{prop:apply 2-level
    cut}. The $i$-level cuts are given by computing the $1$-level cut with the
  proof of Prop.~\ref{prop:max 1-level complexity} as described in
  Section~\ref{sec:cut:adiar:reduce} and then applying Lem.~\ref{prop:2-level
    from 1-level cut:2} to obtain a bound on the 2-level cut.

\item \textbf{$2$-level:} Extends the $1$-level variant by computing $2$-level
  cuts directly with the algorithm based on the proof of Prop.~\ref{prop:max
    2-level complexity} in Section~\ref{sec:cut:adiar:reduce}.
\end{itemize}
All three variants include the computation of $1$-level cuts -- even the \#nodes
one. This reduces the number of variables in our measurements. We have
separately measured the slowdown introduced by computing $1$-level cuts to be
$1.0\%$.

\subsection{Benchmarks} \label{sec:experiments:benchmarks}

We have evaluated the quality of our modifications on the four benchmarks below
that are publicly available at \cite{Soelvsten:bdd-benchmark}. These were
also used to measure the performance of Adiar~1.0 (BDDs) and 1.1 (ZDDs) in
\cite{Soelvsten2022:TACAS, Soelvsten2023:NFM}. The first benchmark is a circuit
verification problem and the others are combinatorial problems.
\begin{itemize}
\item \textbf{\emph{EPFL} Combinational Benchmark Suite~\cite{Amaru2015}.} The
  task is to check equivalence between an original hardware circuit
  (specification) and an optimised circuit (implementation). We construct BDDs
  for all output gates in both circuits, and check if they are equivalent. We
  focus on the $23$ out of the $46$ optimised circuits that Adiar could verify
  in \cite{Soelvsten2022:TACAS}

  Input gates are encoded as a single variable, $x_i$, with a maximum $2$-level
  cut of size $2$.

\item \textbf{Knight's Tour.} On an $N_r \times N_c$ chessboard, the set of all
  paths of a Knight is created by intersecting the valid transitions for each of
  the $N_r N_c$ time steps. The cut of each such ZDD constraint is ${\sim}8N_r
  N_c$. Then, each Hamiltonian constraint with cut size $4$ is imposed onto this
  set \cite{Soelvsten2023:NFM}.

\item \textbf{$N$-Queens.} On an $N \times N$ chessboard, the constraints on
  placing queens are combined per row, based on a base case for each cell. Each
  row constraint is finally accumulated into the complete solution
  \cite{Kunkle2010}.

  For BDDs, each basic cell constraint has a cut size of ${\sim}3N$, while for
  ZDDs it is only $3$.

\item \textbf{Tic-Tac-Toe.} Initially, a BDD or ZDD with cut size ${\sim}N$ is
  created to represent that $N$ crosses have been set within a $4 \times 4
  \times 4$ cube. Then for each of the $76$ lines, a constraint is added to
  exclude any \emph{non-draw} states \cite{Kunkle2010}.

  Each such line constraint has a cut size of $4$ with BDDs and $6$ with ZDDs.
\end{itemize}

\subsection{Tradeoff between Precision and Running Time} \label{sec:experiments:best version}

% Number of benchmark instances run:
% 18 picotrav
% + 10 Open Tours
% + 7 Closed Tours
% + 11 Queens (BDD)
% + 11 Queens (ZDD)
% + 7 Tic-Tac-Toe (BDD)
% + 7 Tic-Tac-Toe (ZDD)
%
% Total: 71 instances
We have run all benchmarks on a consumer-grade laptop with one 2.6 GHz Intel
i7-4720HQ processor, $8$~GiB of RAM, $230$~GiB of available SSD disk, running
Fedora 36, and compiling code with GCC~12.2.1. For each of these $71$ benchmark
instances, Adiar has been given $128$~MiB or $4$~GiB of internal memory.

All combinatorial benchmarks use a unary operation at the end to count the
number of solutions. Table~\ref{tab:count_precision} shows the average ratio
between the predicted and actual maximum size of this operation's priority
queue. As instances grow larger, the quality of the \#nodes heuristic
deteriorates for BDDs. On the other hand, the $1$ and $2$-level cut heuristics
are at most off by a factor of $2$. Hence, since the priority queue's maximum
size is some $2$-level cut, the algorithms in Section~\ref{sec:cut:adiar:reduce}
are only over-approximating the actual maximum $2$-level cut by a factor of $2$.
The result of this is that $i$-level cuts can safely identify that a BDD with
$5.2 \cdot 10^7$ nodes ($1.1$~GiB) can be processed purely within $128$~MiB of
internal memory available. The precision of $i$-level cuts are worse for ZDDs,
but still allow processing a ZDD with $4.3 \cdot 10^7$ nodes ($978$~MiB) with
$128$~MiB of memory.

\if\arxiv1
\begin{table}[b]
\else
\begin{table}[t]
\fi
  \centering

  \caption{Geometric mean of the ratio between the predicted and the actual
    maximum size of the unary Count operation's priority queue. This average is
    also weighed by the input size to gauge the predictions' quality for larger
    BDDs.}
  \label{tab:count_precision}

  \if\arxiv1
    \begin{tabular}{r || c|c|c}
      \multicolumn{4}{c}{BDD}
      \\ \hline \hline
      & \ \#nodes \ & \ $1$-level \ & \ $2$-level
      \\ \hline
      Unweighted Avg.
      &
        2.1\%         & 69.2\%        & 86.3\%
      \\
      Weighted Avg.
      &
        0.1\%         & 76.5\%        & 77.4\%
      \\ \hline
      \multicolumn{4}{c}{ZDD}
      \\ \hline \hline
      Unweighted Avg.
      &
        15.2\%  & 47.8\%    & 67.0\%
      \\
      Weighted Avg.
      &
        25.0\%  & 50.7\%    & 61.8\%
      \\ \hline
    \end{tabular}
  \else
    \begin{tabular}{r || c|c|c || c|c|c}
      & \multicolumn{3}{c||}{BDD}         & \multicolumn{3}{c}{ZDD}
      \\ \hline \hline
      & \ \#nodes \ & \ $1$-level \ & \ $2$-level \ & \ \#nodes \ & \ $1$-level \ & \  $2$-level \
      \\ \hline
      Unweighted Avg.
      &
        2.1\%         & 69.2\%        & 86.3\%    & 15.2\%  & 47.8\%    & 67.0\%
      \\
      Weighted Avg.
      &
        0.1\%         & 76.5\%        & 77.4\%    & 25.0\%  & 50.7\%    & 61.8\%
    \end{tabular}
  \fi
\end{table}
\if\arxiv1
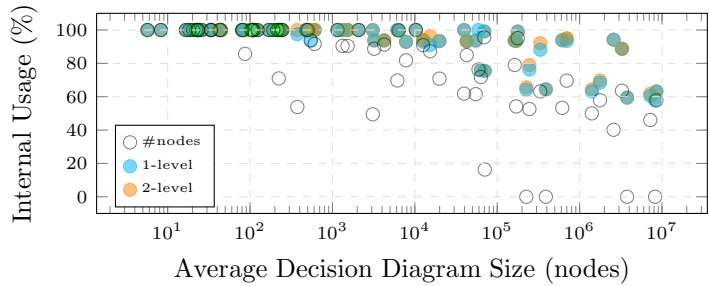
\begin{figure}[t]
\else
\begin{figure}[t]
\fi
  \centering

  \begin{tikzpicture}
    \begin{axis}[%
      width=0.8\linewidth, height=0.35\linewidth,
      every tick label/.append style={font=\scriptsize},
      % x-axis
      xlabel={Average Decision Diagram Size (nodes)},
      xmode=log,
      xmajorgrids=true,
      % y-axis
      ylabel={Internal Usage (\%)},
      ymin=-10,
      ymax=110,
      ytick={0,20,40,60,80,100},
      yminorgrids=false,
      ymajorgrids=true,
      grid style={dashed,black!12},
      % legend
      legend pos=south west,
      legend cell align=left,
      ]

      \begin{scope}[blend mode=soft light]
        % nodes
        \addplot+ [forget plot, style=dots_0nodes]
        table {./data/ir_picotrav_0nodes.tex};
        \addplot+ [forget plot, style=dots_0nodes]
        table {./data/ir_knights_tour_open_0nodes.tex};
        \addplot+ [forget plot, style=dots_0nodes]
        table {./data/ir_knights_tour_closed_0nodes.tex};
        \addplot+ [forget plot, style=dots_0nodes]
        table {./data/ir_queens_bdd_0nodes.tex};
        \addplot+ [forget plot, style=dots_0nodes]
        table {./data/ir_queens_zdd_0nodes.tex};
        \addplot+ [forget plot, style=dots_0nodes]
        table {./data/ir_tic_tac_toe_bdd_0nodes.tex};
        \addplot+ [forget plot, style=dots_0nodes]
        table {./data/ir_tic_tac_toe_zdd_0nodes.tex};

        % 1-level
        \addplot+ [forget plot, style=dots_1level]
        table {./data/ir_picotrav_1level.tex};
        \addplot+ [forget plot, style=dots_1level]
        table {./data/ir_knights_tour_open_1level.tex};
        \addplot+ [forget plot, style=dots_1level]
        table {./data/ir_knights_tour_closed_1level.tex};
        \addplot+ [forget plot, style=dots_1level]
        table {./data/ir_queens_bdd_1level.tex};
        \addplot+ [forget plot, style=dots_1level]
        table {./data/ir_queens_zdd_1level.tex};
        \addplot+ [forget plot, style=dots_1level]
        table {./data/ir_tic_tac_toe_bdd_1level.tex};
        \addplot+ [forget plot, style=dots_1level]
        table {./data/ir_tic_tac_toe_zdd_1level.tex};

        % 2-level
        \addplot+ [forget plot, style=dots_2level]
        table {./data/ir_picotrav_2level.tex};
        \addplot+ [forget plot, style=dots_2level]
        table {./data/ir_knights_tour_open_2level.tex};
        \addplot+ [forget plot, style=dots_2level]
        table {./data/ir_knights_tour_closed_2level.tex};
        \addplot+ [forget plot, style=dots_2level]
        table {./data/ir_queens_bdd_2level.tex};
        \addplot+ [forget plot, style=dots_2level]
        table {./data/ir_queens_zdd_2level.tex};
        \addplot+ [forget plot, style=dots_2level]
        table {./data/ir_queens_bdd_2level.tex};
        \addplot+ [forget plot, style=dots_2level]
        table {./data/ir_tic_tac_toe_bdd_2level.tex};
        \addplot+ [forget plot, style=dots_2level]
        table {./data/ir_tic_tac_toe_zdd_2level.tex};
      \end{scope}

      % legend
      \begin{scope}[blend mode=normal]
        \addlegendimage{style=dots_0nodes}
        \addlegendentry{\tiny \#nodes}

        \addlegendimage{style=dots_1level}
        \addlegendentry{\tiny $1$-level}

        \addlegendimage{style=dots_2level}
        \addlegendentry{\tiny $2$-level}
      \end{scope}
    \end{axis}
  \end{tikzpicture}

  \caption{\emph{Internal} vs.\ \emph{external} memory usage for product constructions (128~MiB).}
  \label{fig:internal_ratio__128MiB}
\end{figure}

This difference in precision affects the product construction algorithms, e.g.\
the Apply operation. Fig.~\ref{fig:internal_ratio__128MiB} shows the amount of
product constructions that each heuristic enables to run with internal memory
data structures. Even when the average BDD was $10^7$ nodes ($229$~MiB) or
larger, with $i$-level cuts at least $59.5\%$ of all algorithms were run purely
in $128$~MiB of memory, whereas with \#nodes sometimes none of them were. Yet,
while there is a major difference between \#nodes and $1$-level cuts, going
further to $2$-level cuts only has a minor effect.

\if\arxiv1
\begin{figure}[t]
\else
\begin{figure}[t]
\fi
  \centering

  \subfloat[$128$~MiB internal memory]{
    \begin{tikzpicture}
      \begin{axis}[%
        width=0.45\linewidth, height=0.35\linewidth,
        every tick label/.append style={font=\scriptsize},
        % x-axis
        xlabel={\scriptsize Maximum Diagram Size},
        xmin=1,
        xtick={1,100,10000,1000000,100000000},
        xmax=1000000000,
        xmode=log,
        % y-axis
        ylabel={\scriptsize Time Difference (\%)},
        ymin=-35,
        ymax=15,
        ytick={-30,-20,-10,0,10},
        % grid
        grid style={dashed,black!12},
        ]

        % scatter line: equal
        \addplot[domain=1:1000000000, samples=8, color=black]
        {0};

        \begin{scope}[blend mode=soft light]
          % average
          \addplot[domain=1:1000000000, samples=8, style=plot_1level]
          {-4.85};

          \addplot[domain=1:1000000000, samples=8, style=plot_2level]
          {-2.56};

          % picotrav
          \addplot+ [forget plot, style=dots_1level]
          table {./data/time_ratio_picotrav_1level__128MiB.tex};
          \addplot+ [forget plot, style=dots_2level]
          table {./data/time_ratio_picotrav_2level__128MiB.tex};

          % knights_tour (open)
          \addplot+ [forget plot, style=dots_1level]
          table {./data/time_ratio_knights_tour_open_1level__128MiB.tex};
          \addplot+ [forget plot, style=dots_2level]
          table {./data/time_ratio_knights_tour_open_2level__128MiB.tex};

          % knights_tour (closed)
          \addplot+ [forget plot, style=dots_1level]
          table {./data/time_ratio_knights_tour_closed_1level__128MiB.tex};
          \addplot+ [forget plot, style=dots_2level]
          table {./data/time_ratio_knights_tour_closed_2level__128MiB.tex};

          % queens_bdd
          \addplot+ [forget plot, style=dots_1level]
          table {./data/time_ratio_queens_bdd_1level__128MiB.tex};
          \addplot+ [forget plot, style=dots_2level]
          table {./data/time_ratio_queens_bdd_2level__128MiB.tex};

          % queens_zdd
          \addplot+ [forget plot, style=dots_1level]
          table {./data/time_ratio_queens_zdd_1level__128MiB.tex};
          \addplot+ [forget plot, style=dots_2level]
          table {./data/time_ratio_queens_zdd_2level__128MiB.tex};

          % tic_tac_toe_bdd
          \addplot+ [forget plot, style=dots_1level]
          table {./data/time_ratio_tic_tac_toe_bdd_1level__128MiB.tex};
          \addplot+ [forget plot, style=dots_2level]
          table {./data/time_ratio_tic_tac_toe_bdd_2level__128MiB.tex};

          % tic_tac_toe_zdd
          \addplot+ [forget plot, style=dots_1level]
          table {./data/time_ratio_tic_tac_toe_zdd_1level__128MiB.tex};
          \addplot+ [forget plot, style=dots_2level]
          table {./data/time_ratio_tic_tac_toe_zdd_2level__128MiB.tex};
        \end{scope}
      \end{axis}
    \end{tikzpicture}
  }
  \quad
  \subfloat[$4$~GiB internal memory]{
    \begin{tikzpicture}
      \begin{axis}[%
        width=0.45\linewidth, height=0.35\linewidth,
        every tick label/.append style={font=\scriptsize},
        % x-axis
        xlabel={\scriptsize Maximum Diagram Size},
        xmin=1,
        xmax=1000000000,
        xmode=log,
        xtick={1,100,10000,1000000,100000000},
        % y-axis
        ylabel={\scriptsize Time Difference (\%)},
        ymin=-35,
        ymax=15,
        ytick={-30,-20,-10,0,10},
        % grid
        grid style={dashed,black!12},
        ]

        % scatter line: equal
        \addplot[domain=1:1000000000, samples=8, color=black]
        {0};

        \begin{scope}[blend mode=soft light]
          % average
          \addplot[domain=1:1000000000, samples=8, style=plot_1level]
          {-4.85};

          \addplot[domain=1:1000000000, samples=8, style=plot_2level]
          {-2.56};

          % picotrav
          \addplot+ [forget plot, style=dots_1level]
          table {./data/time_ratio_picotrav_1level__4GiB.tex};
          \addplot+ [forget plot, style=dots_2level]
          table {./data/time_ratio_picotrav_2level__4GiB.tex};

          % knights_tour (open)
          \addplot+ [forget plot, style=dots_1level]
          table {./data/time_ratio_knights_tour_open_1level__4GiB.tex};
          \addplot+ [forget plot, style=dots_2level]
          table {./data/time_ratio_knights_tour_open_2level__4GiB.tex};

          % knights_tour (closed)
          \addplot+ [forget plot, style=dots_1level]
          table {./data/time_ratio_knights_tour_closed_1level__4GiB.tex};
          \addplot+ [forget plot, style=dots_2level]
          table {./data/time_ratio_knights_tour_closed_2level__4GiB.tex};

          % queens_bdd
          \addplot+ [forget plot, style=dots_1level]
          table {./data/time_ratio_queens_bdd_1level__4GiB.tex};
          \addplot+ [forget plot, style=dots_2level]
          table {./data/time_ratio_queens_bdd_2level__4GiB.tex};

          % queens_zdd
          \addplot+ [forget plot, style=dots_1level]
          table {./data/time_ratio_queens_zdd_1level__4GiB.tex};
          \addplot+ [forget plot, style=dots_2level]
          table {./data/time_ratio_queens_zdd_2level__4GiB.tex};

          % tic_tac_toe_bdd
          \addplot+ [forget plot, style=dots_1level]
          table {./data/time_ratio_tic_tac_toe_bdd_1level__4GiB.tex};
          \addplot+ [forget plot, style=dots_2level]
          table {./data/time_ratio_tic_tac_toe_bdd_2level__4GiB.tex};

          % tic_tac_toe_zdd
          \addplot+ [forget plot, style=dots_1level]
          table {./data/time_ratio_tic_tac_toe_zdd_1level__4GiB.tex};
          \addplot+ [forget plot, style=dots_2level]
          table {./data/time_ratio_tic_tac_toe_zdd_2level__4GiB.tex};
        \end{scope}
      \end{axis}
    \end{tikzpicture}
  }

  \begin{tikzpicture}
    \begin{customlegend}[
      legend columns=-1,
      legend style={draw=none,column sep=1ex},
      legend entries={\footnotesize $1$-level,\footnotesize $2$-level}
      ]
      \addlegendimage{style=dots_1level}
      \addlegendimage{style=dots_2level}
    \end{customlegend}
  \end{tikzpicture}

  \caption{\Adiar\ with $i$-level cuts compared to \#nodes (lower is better).
    Horizontal lines show the average difference in performance.}
  \label{fig:scatter_vs_nodes}
\end{figure}
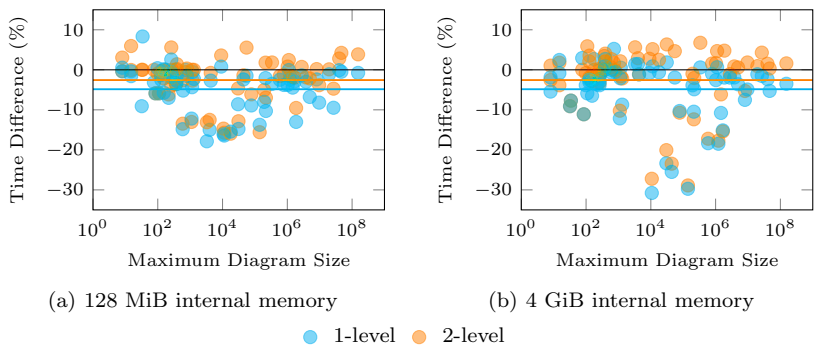

How often internal memory could be used is also reflected in Adiar's
performance. Fig.~\ref{fig:scatter_vs_nodes} shows the difference in the running
time between using $i$-level cuts and only using \#nodes. All benchmarks runs
were interleaved and repeated at least $8$ times. The minimum measured running
time is reported as it minimises any noise due to hardware and the operating
system \cite{Chen2016}. Since the \#nodes version also includes the computation
for the $1$-level cuts but does not use them, any performance decrease in
Fig.~\ref{fig:scatter_vs_nodes} for $1$-level cuts is due to noise.

Using the geometric mean, $1$-level cuts provide a $4.9\%$ improvement over
\#nodes. Considering the $1.0\%$ overhead for computing the $1$-level cuts, this
is a net improvement of $3.9\%$. More importantly, in a considerable amount of
benchmarks, using $i$-level cuts improves the performance by more than $10\%$,
sometimes by $30\%$. These are the benchmark instances where only $i$-level cuts
can guarantee that all auxiliary data structures can fit within internal memory,
yet the instances are still so small that there is a major overhead in
initialising \TPIE's external memory data structures.

The improvement in precision obtained by using $2$-level cuts does not pay off
in comparison to using $1$-level cuts. On average, using $2$-level cuts only
improves the performance of using \#nodes with $2.6\%$. That is, the additional
cost of computing $2$-level cuts outweighs the benefits of its added precision.

Adiar with $i$-level cuts did not slow down as internal memory was increased
from $128$~MiB to $4$~GiB. That is, the precision of both these bounds -- unlike
\#nodes -- ensures that external memory data structures are only used when their
initialisation cost is negligible.
% Furthermore, the share of the faster internal memory data structures
% increases, offsetting the increased initialisation cost of the few remaining
% external memory variants.
Hence, Adiar with $1$-level cuts covers all our needs at the minimal
computational cost and so is included in Adiar~1.2.

\subsection{Impact of Introducing Cuts on Adiar's Running Time} \label{sec:experiments:impact}

In \cite{Soelvsten2022:TACAS,Soelvsten2023:NFM} we measured the performance of
Adiar~1.0 and 1.1 against the conventional BDD packages
CUDD~3.0~\cite{Somenzi2015} and Sylvan~1.5~\cite{Dijk2016}. In those
experiments~\cite{Soelvsten2022:TACAS:Artifact,Soelvsten2023:NFM:Artifact},
Sylvan was not using multi-threading and all experiments were run on machines
with $384$~GiB of RAM of which $300$~GiB was given to the BDD package. To gauge
the impact of using cuts, we now compare our previous measurements without cuts
to new ones with cuts on the exact same hardware and settings. The results of
our new measurements are available at \cite{Soelvsten2023:ATVA:Artifact}.

With $300$~GiB internal memory available, all three modified versions of Adiar
essentially behave the same. Hence, in Fig.~\ref{fig:motivation}
(cf.~Section~\ref{sec:introduction}) we show the best performance for all three
versions on top of the data reported in \cite{Soelvsten2022:TACAS}. Even on the
largest benchmarks we see a performance increase by exploiting cuts. Most
important is the increase in performance for the moderate-size instances where
the initialisation of \TPIE's external memory data structures are costly, e.g.\
$N$-Queens with $N < 11$ and Tic-Tac-Toe with $N < 19$. Based on the data in
\cite{Soelvsten2022:TACAS,Soelvsten2023:NFM} these instances of the
combinatorial benchmarks are the ones where the largest constructed BDD or ZDD
is smaller than $4.9 \cdot 10^6$ nodes ($113$~MiB).

Using the geometric mean, the time spent solving both the combinatorial and
verification benchmarks decreased with Adiar~1.2 on average by $86.1\%$ (with
median $89.7\%$) in comparison to previous versions. For some instances this
difference is even $99.9\%$. In fact, Adiar~1.2 is in some specific instances of
the Tic-Tac-Toe benchmarks faster than CUDD. These are the very instances that
are large enough for CUDD's first -- and comparatively expensive -- garbage
collection to kick in and dominate its running time.

Verifying the EPFL benchmarks involves constructing a few BDDs that are larger
than the $113$~MiB bound mentioned above, but most BDDs are much smaller. As
shown in Table~\ref{tab:picotrav_time}, for the $16$ EPFL circuits\footnote{In
  \cite{Soelvsten2023:ATVA} we incorrectly reported this as being 15 rather than
  16 circuits of this size.} that only generate BDDs smaller than $113$~MiB,
using cuts decreases the computation time on average by $92\%$ (with median
$92\%$). While Adiar~v1.0 still took $56.5$ hours to verify these $16$ circuits,
now with Adiar~1.2 it only takes $4.0$ hours to do the same. These $52.5$ hours
are primarily saved within one of the $16$ circuits. Specifically, using cuts
has decreased the time to verify the \texttt{sin} circuit optimised for depth by
$52.1$ hours. Here, the average BDD size is $2.9$~KiB, the largest BDD
constructed is $25.5$~MiB in size, and up to $42,462$ BDDs are in use
concurrently.

\begin{table}[ht]
  \centering

  \caption{Minimum Running Time to construct the EPFL benchmark
    circuits~\cite{Amaru2015} optimized for depth (d) or size (s) together with
    its respective specification circuit. The variable order $\pi$ was either
    set to be based on a level/depth-first (LD) traversal of the circuit or the
    given input-order (I). Some timings are not provided due to a Memory Out
    (MO), a Time Out (TO), or the measurement has not been made (--). All
    timings for \Adiar~v1.0, CUDD, and Sylvan are from
    \cite{Soelvsten2022:TACAS}, except for \texttt{mem\_ctrl} and \texttt{voter}
    with LD variable ordering.}
  \label{tab:picotrav_time}

  \tiny
  \begin{tabular}{ccc||cc||c|c|c|c}
                       &       &
    \\
    \multicolumn{3}{c||}{Circuit}    & \multicolumn{2}{c||}{BDD Size (MiB)} &  \Adiar~v1.0  &  \Adiar~v1.2  &  CUDD   & Sylvan
    \\
    name               & opt.  & $\pi$ & ~~Avg.~~        & ~~Max~~            & Time (ms)         & Time (ms)         & Time (ms)    & Time (ms)
    \\ \hline \hline
    \texttt{adder}     & d     & LD    & 0.0028      & 0.0182         & 193552        & 8022          & 790     & 170
    \\
                       & s     & LD    & 0.0030      & 0.0088         & 138116        & 5300          & 772     & 91
    \\
    \texttt{arbiter}   & d     & I     & 0.0125      & 63.64          & 472784        & 73638         & 7664    & 24769
    \\
    \texttt{cavlc}     & d+s   & I     & 0.0001      & 0.0022         & 29550         & 1943          & 2       & 8
    \\
    \texttt{ctrl}      & d+s   & I     & 0.0000      & 0.0003         & 5173          & 461           & 0       & 2
    \\
    \texttt{dec}       & d+s   & I     & 0.0001      & 0.0002         & 21305         & 1544          & 0       & 4
    \\
    \texttt{i2c}       & d     & I     & 0.0001      & 0.0060         & 36637         & 2942          & 3       & 9
    \\
                       & s     & I     & 0.0001      & 0.0060         & 36192         & 3290          & 3       & 9
    \\
    \texttt{int2float} & d     & I     & 0.0001      & 0.0035         & 8166          & 783           & 0       & 3
    \\
                       & s     & I     & 0.0001      & 0.0035         & 15205         & 1224          & 0       & 4
    \\
    \texttt{mem\_ctrl} & d     & I     & 3.9550      & 16571          & 400464754     & 357042302     & MO      & TO
    \\
                       & s     & I     & 3.9226      & 16571          & 398777513     & 356500951     & MO      & TO
    \\
                       & d     & LD    & 0.0713      & 34.94          & --            & 199999        & 57728   & --
    \\
                       & s     & LD    & 0.1264      & 264.1          & --            & 298615        & 97441   & --
    \\
    \texttt{priority}  & d     & I     & 0.0001      & 0.0029         & 30864         & 1861          & 2       & 10
    \\
                       & s     & I     & 0.0003      & 0.0049         & 33321         & 2035          & 3       & 11
    \\
    \texttt{router}    & d     & I     & 0.0001      & 0.0073         & 7526          & 545           & 0       & 3
    \\
                       & s     & I     & 0.0001      & 0.0029         & 5644          & 569           & 0       & 2
    \\
    \texttt{sin}       & d     & LD    & 0.0021      & 25.50          & 199739354     & 12268821      & 299403  & 226713
    \\
                       & s     & LD    & 0.8787      & 25.43          & 2585946       & 1840623       & 465770  & 394675
    \\
    \texttt{voter}     & d     & I     & 2.190       & 8241           & 25078751      & 16357661      & MO      & 11191333
    \\
                       & s     & I     & 0.4801      & 8241           & 8520173       & 5197230       & 2307858 & 2775903
    \\
                       & d     & LD    & 1.044       & 4348           & --            & 32637944      & 3950294 & --
    \\
                       & s     & LD    & 0.2477      & 249            & --            & 1391096       & 295991  & --
  \end{tabular}
\end{table}

Despite this massive performance improvement with Adiar~1.2 due to our new
technique, there is still a significant gap of 3.7 hours with CUDD and Sylvan on
these 16 circuits. We attribute this to the fact that these benchmarks also
include many computations on really tiny BDDs. Although we keep the auxiliary
data structures in internal memory, the resulting BDDs are still stored on disk,
even when they consist of only a few nodes.

\section{Conclusion} \label{sec:conclusion}

We introduce the idea of a maximum $i$-level cut for DAGs that restricts the cut
to be within a certain window. For $i \in \{ 1,2 \}$ the problem of computing
the maximum $i$-level cut is polynomial-time computable. But, we have been able
to piggyback a slight over-approximation with only a $1\%$ linear overhead onto
Adiar's I/O-efficient bottom-up Reduce operation.

An $i$-level cut captures the shape of Adiar's auxiliary data structures during
the execution of its I/O-efficient time-forward processing algorithms. Hence,
similar to how conventional recursive BDD algorithms have the size of their call
stack linearly dependent on the depth of the input, the maximum $2$-level cuts
provide a sound upper bound on the memory used during Adiar's computation. Using
this, Adiar~1.2 can deduce soundly whether using exclusively internal memory is
possible, increasing its performance in those cases. Doing so decreases
computation time for moderate-size instances up to $99.9\%$ and on average by
$86.1\%$ (with median $89.7\%$).

\subsection{Related and Future Work}

Many approaches tried to achieve large-scale BDD manipulation with distributed
memory algorithms, some based on breadth-first algorithms, e.g.\
\cite{Stornetta1996,Yang1997,Milvang-Jensen1998,Kunkle2010}. Yet, none of these
approaches obtained a satisfactory performance. The speedup obtained by a
multicore implementation~\cite{Dijk2016} relies on parallel depth-first
algorithms using concurrent hash tables, which doesn't scale to external memory.

CAL~\cite{Sanghavi1996} (based on a breadth-first approach
\cite{Ochi1993,Ashar1994}) is to the best of our knowledge the only other BDD
package designed to process large BDDs on a single machine. CAL is I/O
efficient, assuming that a single BDD level fits into main memory; the I/O
efficiency of Adiar does not depend on this assumption. Similar to Adiar, CAL
suffers from bad performance for small instances. To deal with this, CAL
switches to the classical recursive depth-first algorithms when all the given
input BDDs contain fewer than $2^{19}$ nodes ($15$~MiB). As far as we can tell,
CAL's threshold is purely based on experimental results of performance and
without any guarantees of soundness. That is, the output may potentially exceed
main memory despite all inputs being smaller than $2^{19}$ nodes, which would
slow it down significantly due to random-access. For BDDs smaller than CAL's
threshold of $2^{19}$ nodes, Adiar~1.2 with $i$-level cuts could run almost all
of our experiments with auxiliary data structures purely in internal memory.

Yet, as is evident in Fig.~\ref{fig:motivation}, when dealing with decision
diagrams smaller than 44.000 nodes (1~MiB), there is still a considerable gap
between Adiar's performance and conventional depth-first based BDD packages (see
also end of Sec.~\ref{sec:experiments:impact}). Apparently, we have reached a
lower bound on the BDD size for which time-forward processing on external memory
is efficient. Solving this would require an entirely different approach: one
that can efficiently and seamlessly combine BDDs stored in internal memory with
BDDs stored in external memory.

\subsection{Applicability Beyond Decision Diagrams}

Our idea is generalisable to all time-forward processing algorithms: the
contents of the priority queues are at any point in time a $2$-level cut with
respect to the input and/or output DAG. Hence, one can bound the algorithm's
memory usage if one can compute a levelisation function and the $1$-level cuts
of the inputs.

A levelisation function is derivable with the preprocessing step in
\cite{Hellings2012} and the cut sizes can be computed with an I/O-efficient
version of the greedy algorithm presented in this paper. Yet for our approach to
be useful in practice, one has to identify a levelisation function that best
captures the structure of the DAG in relation to the succeeding algorithms and
where both the computation of the levelisation and the $1$-level cut can be
computed with only a negligible overhead -- preferably within the other
algorithms.

\clearpage
\section*{Acknowledgements}

We want to thank Anna Blume Jakobsen for her help implementing the use of
$i$-level cuts in Adiar and Kristoffer Arnsfelt Hansen for his input on the
computational complexity of these cuts. Finally, thanks to the Centre for
Scientific Computing, Aarhus,
(\href{http://phys.au.dk/forskning/cscaa/}{phys.au.dk/forskning/cscaa/}) for
running our benchmarks.

\bibliographystyle{plain}
\bibliography{references}

\begin{thebibliography}{10}

\bibitem{Aggarwal1987}
Alok Aggarwal and S.~Vitter, Jeffrey.
\newblock The input/output complexity of sorting and related problems.
\newblock {\em Communications of the ACM}, 31(9):1116--1127, 1988.

\bibitem{Amaru2015}
Luca Amar{\'u}, Pierre-Emmanuel Gaillardon, and Giovanni De~Micheli.
\newblock The {EPFL} combinational benchmark suite.
\newblock In {\em 24th International Workshop on Logic \& Synthesis}, 2015.

\bibitem{Amparore2022}
Elvio~Gilberto Amparore, Susanna Donatelli, and Francesco Gall{\`{a}}.
\newblock {starMC}: an automata based {CTL*} model checker.
\newblock {\em PeerJ Comput. Sci.}, 8:e823, 2022.

\bibitem{Arge1995}
Lars Arge.
\newblock The {I/O}-complexity of ordered binary-decision diagram manipulation.
\newblock In {\em 6th International Symposium on Algorithms and Computations
  (ISAAC)}, volume 1004 of {\em Lecture Notes in Computer Science}, pages
  82--91, 1995.

\bibitem{Arge1996}
Lars Arge.
\newblock The {I/O}-complexity of ordered binary-decision diagram.
\newblock In {\em BRICS RS preprint series}, volume~29. Department of Computer
  Science, University of Aarhus, 1996.

\bibitem{Ashar1994}
Pranav Ashar and Matthew Cheong.
\newblock Efficient breadth-first manipulation of binary decision diagrams.
\newblock In {\em IEEE/ACM International Conference on Computer-Aided Design
  (ICCAD)}, pages 622--627. IEEE Computer Society Press, 1994.

\bibitem{Brace1990}
Karl~S. Brace, Richard~L. Rudell, and Randal~E. Bryant.
\newblock Efficient implementation of a {BDD} package.
\newblock In {\em 27th Design Automation Conference (DAC)}, pages 40--45.
  Association for Computing Machinery, 1990.

\bibitem{Bryant1986}
Randal~E. Bryant.
\newblock Graph-based algorithms for {B}oolean function manipulation.
\newblock {\em IEEE Transactions on Computers}, C-35(8):677--691, 1986.

\bibitem{Bryant2022}
Randal~E. Bryant, Armin Biere, and Marijn J.~H. Heule.
\newblock Clausal proofs for pseudo-{B}oolean reasoning.
\newblock In {\em Tools and Algorithms for the Construction and Analysis of
  Systems}, pages 443--461. Springer, 2022.

\bibitem{Bryant2021:QBF}
Randal~E. Bryant and Marijn J.~H. Heule.
\newblock Dual proof generation for quantified {Boolean} formulas with a
  {BDD}-based solver.
\newblock In {\em Automated Deduction -- CADE 28}, pages 433--449. Springer,
  2021.

\bibitem{Bryant2021:SAT}
Randal~E. Bryant and Marijn J.~H. Heule.
\newblock Generating extended resolution proofs with a {BDD}-based sat solver.
\newblock In {\em Tools and Algorithms for the Construction and Analysis of
  Systems (TACAS)}, volume 12651 of {\em Lecture Notes in Computer Science},
  pages 76--93. Springer, 2021.

\bibitem{Chen2016}
Jiahao Chen and Jarrett Revels.
\newblock Robust benchmarking in noisy environments.
\newblock arXiv, 2016.

\bibitem{Chiang1995}
Yi-Jen Chiang, Michael~T. Goodrich, Edward~F. Grove, Roberto Tamassia,
  Darren~Erik Vengroff, and Jeffrey~Scott Vitter.
\newblock External-memory graph algorithms.
\newblock In {\em Proceedings of the Sixth Annual ACM-SIAM Symposium on
  Discrete Algorithms}, SODA '95, pages 139–--149. Society for Industrial and
  Applied Mathematics, 1995.

\bibitem{Ciardo2009}
Gianfranco Ciardo, Andrew~S. Miner, and Min Wan.
\newblock Advanced features in {SMART:} the stochastic model checking analyzer
  for reliability and timing.
\newblock {\em {SIGMETRICS} Perform. Evaluation Rev.}, 36(4):58--63, 2009.

\bibitem{Cimatti2000}
Alessandro Cimatti, Edmund Clarke, Fausto Giunchiglia, and Marco Roveri.
\newblock {N}u{SMV}: a new symbolic model checker.
\newblock {\em International Journal on Software Tools for Technology
  Transfer}, 2:410--425, 2000.

\bibitem{Dijk2016}
Tom {\noopsort{Dijk}{Van Dijk}} and Jaco {\noopsort{Pol}{Van de Pol}}.
\newblock Sylvan: multi-core framework for decision diagrams.
\newblock {\em International Journal on Software Tools for Technology
  Transfer}, 19:675--696, 2016.

\bibitem{Gammie2004}
Peter Gammie and Ron {\noopsort{Meyden}{Van der Meyden}}.
\newblock {MCK}: Model checking the logic of knowledge.
\newblock In {\em Computer Aided Verification}, volume 3114 of {\em Lecture
  Notes in Computer Science}, pages 479--483, Berlin, Heidelberg, 2004.
  Springer.

\bibitem{He2020}
Leifeng He and Guanjun Liu.
\newblock {P}etri net based symbolic model checking for computation tree logic
  of knowledge.
\newblock arXiv, 2020.

\bibitem{Hellings2012}
Jelle Hellings, George~H.L. Fletcher, and Herman Haverkort.
\newblock Efficient external-memory bisimulation on {DAGs}.
\newblock In {\em Proceedings of the 2012 ACM SIGMOD International Conference
  on Management of Data}, SIGMOD '12, pages 553–--564. Association for
  Computing Machinery, 2012.

\bibitem{Kant2015}
Gijs Kant, Alfons Laarman, Jeroenn Meijer, Jaco {\noopsort{Pol}{Van de Pol}},
  Stefan Blom, and Tom {\noopsort{Dijk}{Van Dijk}}.
\newblock {LTS}min: High-performance language-independent model checking.
\newblock In {\em Tools and Algorithms for the Construction and Analysis of
  Systems (TACAS)}, volume 9035 of {\em Lecture Notes in Computer Science},
  pages 692--707, Berlin, Heidelberg, 2015. Springer.

\bibitem{Kunkle2010}
Daniel Kunkle, Vlad Slavici, and Gene Cooperman.
\newblock Parallel disk-based computation for large, monolithic binary decision
  diagrams.
\newblock In {\em 4th International Workshop on Parallel Symbolic Computation
  (PASCO)}, pages 63--72, 2010.

\bibitem{Lampis2011}
Michael Lampis, Georgia Kaouri, and Valia Mitsou.
\newblock On the algorithmic effectiveness of digraph decompositions and
  complexity measures.
\newblock {\em Discrete Optimization}, 8(1):129--138, 2011.
\newblock Parameterized Complexity of Discrete Optimization.

\bibitem{Yi2022}
Yi~Lin, Lucas~M. Tabajara, and Moshe~Y. Vardi.
\newblock {ZDD} {Boolean} synthesis.
\newblock In {\em Tools and Algorithms for the Construction and Analysis of
  Systems}, pages 64--83. Springer, 2022.

\bibitem{Lomuscio2017}
Alessio Lomuscio, Hongyang Qu, and Franco Raimondi.
\newblock {MCMAS}: an open-source model checker for the verification of
  multi-agent systems.
\newblock {\em International Journal on Software Tools for Technology
  Transfer}, 19:9--30, 2017.

\bibitem{Meyer2003}
Ulrich Meyer, Peter Sanders, and Jop Sibeyn.
\newblock {\em Algorithms for Memory Hierarchies: Advanced Lectures}.
\newblock Springer, Berlin, Heidelberg, 2003.

\bibitem{Milvang-Jensen1998}
Kim Milvang-Jensen and Alan~J. Hu.
\newblock {BDDNOW}: a parallel {BDD} package.
\newblock In {\em Formal Methods in Computer-Aided Design (FMCAD)}, pages
  501--507, Berlin, Heidelberg, 1998. Springer.

\bibitem{Minato1993}
Shin-ichi Minato.
\newblock Zero-suppressed {BDD}s for set manipulation in combinatorial
  problems.
\newblock In {\em 30th Design Automation Conference (DAC)}, pages 272--277.
  Association for Computing Machinery, 1993.

\bibitem{Minato1990}
Shin-ichi Minato, Nagisa Ishiura, and Shuzo Yajima.
\newblock Shared binary decision diagram with attributed edges for efficient
  {B}oolean function manipulation.
\newblock In {\em 27th Design Automation Conference (DAC)}, pages 52--57.
  Association for Computing Machinery, 1990.

\bibitem{Molhave2012}
Thomas Mølhave.
\newblock Using \textsc{TPIE} for processing massive data sets in \textsc{C}++.
\newblock Technical report, Duke University, Durham, NC, 2012.

\bibitem{Ochi1993}
Hiroyuki Ochi, Koichi Yasuoka, and Shuzo Yajima.
\newblock Breadth-first manipulation of very large binary-decision diagrams.
\newblock In {\em International Conference on Computer Aided Design (ICCAD)},
  pages 48--55. IEEE Computer Society Press, 1993.

\bibitem{Papadimitriou1991}
Christos~H. Papadimitriou and Mihalis Yannakakis.
\newblock Optimization, approximation, and complexity classes.
\newblock {\em Journal of Computer and System Sciences}, 43(3):425--440, 1991.

\bibitem{Sanghavi1996}
Jagesh~V. Sanghavi, Rajeev~K. Ranjan, Robert~K. Brayton, and Alberto
  Sangiovanni-Vincentelli.
\newblock High performance {BDD} package by exploiting memory hierarchy.
\newblock In {\em 33rd Design Automation Conference (DAC)}, pages 635--640.
  Association for Computing Machinery, 1996.

\bibitem{Soelvsten2023:ATVA}
Steffan~Christ S{\o}lvsten and Jaco {\noopsort{Pol}{van de Pol}}.
\newblock Predicting memory demands of {BDD} operations using maximum graph
  cuts.
\newblock In {\em Automated Technology for Verification and Analysis}, volume
  14216 of {\em Lecture Notes in Computer Science}, pages 72--92. Springer,
  2023.

\bibitem{Somenzi2015}
Fabio Somenzi.
\newblock {CUDD}: {CU} decision diagram package, 3.0.
\newblock Technical report, University of Colorado at Boulder, 2015.

\bibitem{Stornetta1996}
T.~Stornetta and F.~Brewer.
\newblock Implementation of an efficient parallel {BDD} package.
\newblock In {\em Design Automation Conference Proceedings}, volume~33, pages
  641--644, 1996.

\bibitem{Soelvsten:bdd-benchmark}
Steffan~Christ Sølvsten.
\newblock {BDD} {Benchmark}.
\newblock Zenodo, 2022.

\bibitem{Soelvsten2022:TACAS:Artifact}
Steffan~Christ Sølvsten and Jaco {\noopsort{Pol}{van de Pol}}.
\newblock Adiar 1.0.1 : Experiment data, 11 2021.

\bibitem{Soelvsten2023:NFM}
Steffan~Christ Sølvsten and Jaco {\noopsort{Pol}{van de Pol}}.
\newblock {A}diar 1.1: {Z}ero-suppressed {D}ecision {D}iagrams in {E}xternal
  {M}emory.
\newblock In {\em NASA Formal Methods Symposium}, LNCS 13903, Berlin,
  Heidelberg, 2023. Springer.

\bibitem{Soelvsten2023:NFM:Artifact}
Steffan~Christ Sølvsten and Jaco {\noopsort{Pol}{van de Pol}}.
\newblock Adiar 1.1.0 : Experiment data, 03 2023.

\bibitem{Soelvsten2023:ATVA:Artifact}
Steffan~Christ Sølvsten and Jaco {\noopsort{Pol}{van de Pol}}.
\newblock Adiar 1.2.0 : Experiment data, 07 2023.

\bibitem{Soelvsten2022:TACAS}
Steffan~Christ Sølvsten, Jaco {\noopsort{Pol}{van de Pol}}, Anna~Blume
  Jakobsen, and Mathias Weller~Berg Thomasen.
\newblock {A}diar: {B}inary {D}ecision {D}iagrams in {E}xternal {M}emory.
\newblock In {\em Tools and Algorithms for the Construction and Analysis of
  Systems}, volume 13244 of {\em Lecture Notes in Computer Science}, pages
  295--313, Berlin, Heidelberg, 2022. Springer.

\bibitem{Vengroff1994}
Darren~Erik Vengroff.
\newblock A transparent parallel \textsc{I/O} environment.
\newblock In {\em In Proc. 1994 DAGS Symposium on Parallel Computation}, pages
  117--134, 1994.

\bibitem{Yang1997}
Bwolen Yang and David~R. O'Hallaron.
\newblock Parallel breadth-first {BDD} construction.
\newblock {\em SIGPLAN Not.}, 32(7):145--156, 06 1997.

\end{thebibliography}

% ---------------------------------------------------------------------------- %

\end{document}